\documentclass[review,authoryear]{elsarticle}
\pdfoutput=1
\usepackage{pdflscape}
\usepackage{graphicx}
\usepackage{algorithm}
\usepackage{algorithmicx}
\usepackage{listings}
\usepackage{amssymb}
\usepackage{amsmath}
\usepackage{amsfonts}
\usepackage{natbib}
\usepackage{bm}
\usepackage{xr}

\usepackage{mathrsfs}
\usepackage{amscd}
\usepackage[mathscr]{eucal}
\usepackage{extarrows}
\usepackage{color}
\usepackage{nameref}

\usepackage{fixltx2e}
\usepackage{textcomp}
\usepackage{fullpage}
\usepackage{amsfonts}
\usepackage{verbatim}
\usepackage[english]{babel}
\usepackage{pifont}
\usepackage{color}
\usepackage{setspace}
\usepackage{lscape}
\usepackage{indentfirst}
\usepackage[normalem]{ulem}
\usepackage{booktabs}
\usepackage{nag}
\usepackage{float}
\usepackage{latexsym}
\usepackage{url}
\usepackage{hyperref}
\usepackage{epsfig}
\usepackage{array}
\usepackage{ifthen}
\usepackage{amsthm}
\usepackage{amstext}
\usepackage{enumitem}

\usepackage{longtable}
\usepackage[tablesonly, notablist,  nomarkers, tabhead]{endfloat}

\usepackage{footnote}
\usepackage[para]{footmisc}
\usepackage{tikz}
\usepackage{fp}   
\usetikzlibrary{fpu}
\usepackage{nameref}
\usepackage{siunitx}
\usepackage{boxhandler}

\usepackage{xr}

\newcommand{\pkg}[1]{\textbf{#1}}
\newcommand{\proglang}[1]{\texttt{#1}}
\newcommand{\code}[1]{\texttt{\textit{#1}}}

\newcommand{\be}{\begin{equation}} \newcommand{\ee}{\end{equation}}
\newcommand{\bd}{\begin{displaymath}} \newcommand{\ed}{\end{displaymath}}
\newcommand{\ba}{\begin{align}} \newcommand{\ea}{\end{align}}
\newcommand{\baa}{\begin{align*}} \newcommand{\eaa}{\end{align*}}
\newcommand{\ben}{\begin{enumerate}} \newcommand{\een}{\end{enumerate}}
\newcommand{\bi}{\begin{itemize}} \newcommand{\ei}{\end{itemize}}

\newcommand{\ud}{\mathrm{d}}
\newcommand{\E}[1]{\operatorname{E}\left[ #1 \right]}
\newcommand{\Expectation}[1]{\operatorname{E}\left[ #1 \right]}

\newcommand{\var}[1]{\operatorname{Var}\left[ #1 \right]}
\newcommand{\variance}[1]{\operatorname{Var}\left[ #1 \right]}

\newcommand{\KBcom}[1]{}
\newcommand{\KBtext}[1]{\textcolor{black}{#1}}
\newcommand{\TS}[1]{\textcolor{black}{#1}}
\newcommand{\TScom}[1]{}
\newcommand{\VMcom}[1]{}
\newcommand{\VM}[1]{\textcolor{black}{#1}}
\newcommand{\GAcom}[1]{}
\newcommand{\del}[1]{}

\usepackage[normalem]{ulem}

\newtheorem{theorem}{Theorem}
\newtheorem{definition}{Definition}
\newtheorem{remark}{Remark}
\newtheorem{corollary}{Corollary}

\usepackage{bigints}

\journal{}

\begin{document}

\begin{frontmatter}

\title{Fast likelihood evaluation for multivariate phylogenetic comparative methods: the \pkg{PCMBase} \proglang{R} package}

\author[ETH,SIB]{Venelin Mitov\corref{cor1}}
\ead{vmitov@gmail.com}

\author[LiU]{Krzysztof Bartoszek}
\ead{krzysztof.bartoszek@liu.se,bartoszekkj@gmail.com}

\author[ETH]{Georgios Asimomitis}
\ead{georgios.asimomitis@gmail.com}

\author[ETH,SIB]{Tanja Stadler}
\ead{tanja.stadler@bsse.ethz.ch}

\cortext[cor1]{Corresponding author}
\address[ETH]{Department of Biosystems Science and Engineering, Eidgenossische Technische Hochschule Z\"urich, Basel, Switzerland}
\address[SIB]{Swiss Institute for Bioinformatics, Lausanne, Switzerland}
\address[LiU]{  Department of Computer and Information Science, Link\"oping University, Link\"oping, Sweden}

\begin{abstract}

We introduce an \proglang{R} package, \pkg{PCMBase}, 
to rapidly calculate the likelihood for multivariate phylogenetic comparative methods. 
The package is not  specific to  \TS{particular models but offers} the user the functionality to very easily implement a 
wide range of models where the transition along a branch is multivariate normal.
We demonstrate the \TS{package's possibilities on the now standard, multitrait Ornstein--Uhlenbeck 
process as well as the novel multivariate punctuated equilibrium model. The package can handle trees of various
types (e.g.  ultrametric, nonultrametric, polytomies, e.t.c.)}, as well as measurement error, missing measurements 
or non-existing traits for some of the species in the tree. 

\end{abstract}


\end{frontmatter}

\section[Introduction]{Introduction} \label{sec:intro}
Since \citet{Felsenstein:1985bl}'s work describing the independent contrasts algorithm, phylogenetic comparative
methods (PCMs) have steadily been generalized with respect to available models and implementations of them. 
Following \citet{Felsenstein:1988gb}'s suggestion, \citet{Hansen:1997ek} described the Ornstein--Uhlenbeck (OU) process
in the PCM setting. This led to the implementation of OU models in various
packages such as \pkg{ouch} \citep{Butler:2004ce} or \pkg{geiger} \citep{Harmon:2008iy} to name a few,
making it a standard model in the community alongside the Brownian motion (BM) process
popularized in the community by \citet{Felsenstein:1985bl} but see also \citet{Edwards:1970td,LANDE:1976ga}. 
\TS{For species being characterized by multiple traits, 
 the  multivariate OU processes was introduced} by 
\proglang{R} packages such as \pkg{ouch}, \pkg{slouch} \citep{Hansen:2008gt}, \pkg{mvSLOUCH} \citep{Bartoszek:2012fw},
\pkg{mvMORPH} \citep{Clavel:2015hc}, \pkg{Rphylopars} \citep{Goolsby:2016fp}, again, to name a few. 
At the core of these methods, the likelihood of the model parameters and tree for given trait data  is 
evaluated, meaning the probability density of the tip trait values given the parameters and tree is calculated.

From a statistical point of view, 
the development of phylogenetic comparative methods
goes in two directions. The first direction is development of model classes
beyond simple stochastic processes, such as BM and OU,
and the second direction is the development 
of efficient likelihood evaluation methods. 
Considering the first direction, we briefly mention three recent 
proposals. \citet{Manceau:2016kz}
show (with implementation in \pkg{RPANDA}) that if one models 
the suite of traits by a linear stochastic differential equation
\citep[SDE, see the representation by Eq. (1) of][]{Manceau:2016kz}
whose drift matrix (``deterministic part'' of the SDE) is piecewise constant
with respect to the phylogeny, and diffusion matrix 
(``random part'', sometimes referred to as ``random drift part'' in biological literature) 
does not depend on the trait, then the tip measurements are multivariate normal.
The tip measurements' mean vector and covariance matrix can be found by integrating (backwards along the tree)
an appropriate collection of ordinary differential equations (ODEs), which in turn is required for the likelihood calculation. 
\citet{Landis:2012hu, Duchen:2017fm} went beyond the SDE world, represented through the equation,

\be \label{eqSDE}
\ud \vec{X}_{t} = \mu(t,\vec{X}_{t})+\sigma(t,\vec{X}_{t})\ud \vec{W}_{t},
\ee
where $\vec{W}_{t}$ is a standard multivariate Wiener process, 
into L\'evy process models. These are highly relevant from a biological point of view as they
allow for jumps in the trait at random time instances. Hence, they hold promise for attacking
the longstanding question of whether ``evolution is gradual or punctuated?''.
Both approaches consider the transition densities, meaning the change of a trait between the 
start and the end of a branch, when quantifying trait evolution along a phylogeny.
The third approach is 
to model the evolution of the traits' density in time with a partial differential
equation \citep[PDE][]{Boucher:2018hi, Blomberg:2017bd}. 
E.g. in the simplest standard Wiener process case \TS{the PDEs are}
$\frac{\partial}{\partial t}f_{t}(x) = \frac{1}{2}\frac{\partial^{2}}{\partial x^{2}}f_{t}(x)$ 
with boundary condition $f_{0}(x)=\delta_{0}(x)$, i.e. Dirac $\delta$ at $0$.
\VM{\del{This approach is convenient as it is next to impossible to analytically express the 
transition density.}} \TScom{I dont understand this last sentence, eg fitzjohn and you provide the transition 
densities. generally if multivariate normal, we can evaluate the likelihood, as also stated in next paragraph...
thus, delete this sentence?}
\del{ the dynamics of its changes with time are possible 
to describe by means of an appropriate PDE}
\del{Then, some, especially long--term properties (such as stationarity, if it exists), can be obtained  sometimes. }

The other direction is the 
development of efficient likelihood evaluation methods.
Commonly, in PCMs, the model classes have the property that the joint distribution
of the tip measurements is multivariate normal. Hence, there is a closed
form for the likelihood---the multivariate normal density function, 
i.e. an algebraic expression in terms of the traits' mean vector 
and the traits' variance-covariance matrix $(\mathbf{V})$.
Even though it is possible to obtain a conceptually simple equation, actually
calculating the
value of the likelihood is a computational challenge. 
If one has multiple correlated trait measurements per species, 
then 
$\mathbf{V}$ 
can have a very complicated formula \citep[cf. Eqs (A.1, B.3, B.7) of][]{Bartoszek:2012fw}. 
\del{Working with the between--species--between--traits
variance--covariance matrix 
explicitly is a non--trivial exercise. }
As \citet{Freckleton:2012hx} points out ``First, the matrix has to be generated 
in the first place. This requires allocating enough memory to hold all of the entries of
$\mathbf{V}$ and then initiating one traversal (i.e. successively visiting all the nodes)
of the phylogeny per pair of species sharing an ancestor to measure the shared path lengths.
Second $\mathbf{V}$ has to be inverted at one point
in the analysis.''. 

Hence, effort has been invested into 
reducing the memory and time complexity of 
the likelihood evaluation process.
Inspired by \citet{Felsenstein:1973tw}'s approach, \citet{Freckleton:2012hx} proposed a linear 
(w.r.t. the number of tips of the phylogeny) way to
obtain the likelihood for traits evolving as a Brownian motion. 
\citet{Freckleton:2012hx}, further indicates that non--Brownian models can be quickly evaluated
if one appropriately transforms the phylogeny. 
Then, \citet{Ho:2014ge} proposed a general method that takes advantage of the so--called
$3$--point structure of the Brownian motion's between--species--between--traits
variance--covariance matrix
\citep[i.e. a matrix $\mathbf{S}$ has a $3$--point structure if it is symmetric, with 
non--negative entries and for all $i$, $j$, $k$ (possibly equal), 
the two smallest of 
$\mathbf{S}_{ij}$, $\mathbf{S}_{ik}$ and $\mathbf{S}_{jk}$ 
are equal][]{Ho:2014ge}
 and \del{can} obtain the likelihood in linear (w.r.t. the number of tips 
of the phylogeny) time, without having to construct in quadratic time the matrix $\mathbf{V}$.
\TScom{shall we include here the definition of a 3 point structure? If not, we should aslo delete the sentence 
on the definition of the generalized 3-point structure.}
Similarly, calculating the likelihood for
non--Brownian models (like the univariate Ornstein--Uhlenbeck process) can be done in linear time, as long as 
their $\mathbf{V}$ satisfies a generalized 
$3$--point structure.
\TS{Briefly, a covariance matrix satisfies the generalized $3$--point structure if there exist diagonal
matrices $\mathbf{D}_{1}$ and $\mathbf{D}_{2}$ such that $\mathbf{D}_{1}\mathbf{V}\mathbf{D}_{2}$
satisfies the $3$--point structure.} 
\citet{Goolsby:2016fp} \TS{derives such a transformation} to find the likelihood for traits under multivariate 
Ornstein--Uhlenbeck evolution in linear time.
But in their implementation, only ultrametric trees and symmetric--positive--definite drift matrices are supported
at the moment. 
For non--Gaussian models,
a quasi--likelihood is defined and again the same approach (as long as the generalized $3$--point
structure holds) \TS{can be} used \citep{Ho:2014ge}. \TScom{I moved this sentence down, as in my understandign 
Goolsby falls under Gaussian models.}

The speed--up for the Brownian motion's $3$--point structure \TS{(or generalized 3-point structure)} is based on the fact 
that the between--species--between--traits variance--covariance matrix
has a nested structure. Therefore,
appropriate linear algebra allows for rapid calculation of 
$\mathrm{det}(\mathbf{V})$ and quadratic forms like
$\vec{x}\mathbf{V}^{-1}\vec{y}$ \del{, i.e.} without the need to do the inversion $\mathbf{V}^{-1}$.
\del{, 
an $O(p^{2.373})$ operation at least, where $p$ is the dimension of $\mathbf{V}$
} 
\del{Non--Brownian based models, and if the tree is  non--ultrametric (e.g. it has fossil measurements), 
can also take advantage of this method, after appropriate transformations of the phylogeny. }

Even though linear--time likelihood evaluation based on the $3$--point structure is mathematically
elegant, it is, due to the necessity of finding an appropriate transformation for non--Brownian motion,
intrinsically complicated and may seem daunting for a non--algebraically oriented user or developer.
\citet{FitzJohn:2012ey} indicated a probabilistically motivated way of quickly finding the likelihood
(with implementation in the \pkg{Diversitree} \proglang{R} package). \TScom{can we say this approach is also linear
 (as it only traverses each branch once)?}
He noticed (in the Supporting Information), same as \citet{OPybus:2012pnas}, that one can traverse the tree and successively
integrate out the internal nodes. \citet{FitzJohn:2012ey}'s description was focused around the BM and 
univariate OU processes on ultrametric trees. Furthermore, \citet{FitzJohn:2012ey} writes that he proved correctness 
of his method for a three
tip phylogeny and then for larger trees checked numerically.

The presence of two different approaches, namely the 3--point structure method and the tree traversal method,
 to quickly calculating the likelihood combined with a number of 
independent implementations, each with some given set of conditions, can easily cause confusion.
In fact, it seems that this 
led \citet{Slater:2014hy} to write in his Correction (due to ``$\ldots$ errors
arose from use of branch length rescaling under the Ornstein--Uhlenbeck process, which I here show to be 
inappropriate for non--ultrametric trees''), that ``$\ldots$, there is little, if any documentation in the 
literature or elsewhere highlighting that one of these approaches can be used while another cannot.''

In this paper
we attempt to 
\VM{\del{assess} overcome} the difficulties highlighted in the previous paragraph by proposing
a \TS{fast} method to obtain the likelihood \TS{which integrates} over the internal node values.
Our approach is appropriate for a large class of models, namely for all models where
conditional on the ancestral trait, the descendant
trait is normally distributed \VM{\del{(however, we indicate \TS{in the Discussion} that substantial relaxations of 
this are possible)}},
the descendant's expectation
depends linearly on the ancestor, and the variance does not depend on the ancestral value. 
From a mathematical point of view, we provide an inductive proof of \citet{FitzJohn:2012ey}'s claim of method correctness
for multiple traits and all kinds of trees.
 \citet{OPybus:2012pnas} point out that for such a method to work, it is needed
``to keep track of partial'' means and precisions. Here, we propose a very general, computationally
effective, and developer friendly way of doing this by recursively updating the polynomial
representation of the multivariate normal density function.
\TS{In order to use our approach for some new model, one has to be able to calculate} \del{What is required is 
to provide a way to calculate} the variance of the transition along the branch,
the shift in the mean along a branch, and the linear dependency (i.e. a matrix) on the ancestral state. 
\TS{Thus, in} our probabilistic approach, one needs to understand only the dynamics
of a single branch (lineage), 
something that is usually present at the model formulation stage.
For OU based models, these \TS{quantities can be} \del{are} analytically calculated and we provide an implementation. 
For other models, a developer will have to do the calculations themselves, but \del{we believe
that this is} \TS{this should be} significantly less involved than finding the transformation for the 
$3$--point structure. 
In fact, for SDE--type  models, \citet{Manceau:2016kz} provide a general ODE method
(Eqs. S2 and S3) to obtain the conditional mean and variance.
Furthermore, our method can naturally handle measurement error (intra--species variability), missing data, and
punctuated components (jumps), and allows for changes in parameters at arbitrary points along the tree. 
It is \TS{further} appropriate for non--ultrametric,
binary and multifurcating trees. All of such specifications can be provided by the user.
In no case is any tree transformation required. 

Our method 
encompasses a number of contemporary
frameworks. In particular all OU type models (e.g. \pkg{ouch}, \pkg{slouch}, \pkg{mvSLOUCH}, \pkg{mvMORPH})
are covered by it. The \pkg{RPANDA} SDE framework (without interactions between lineages) is also covered
as are current punctuated equilibrium models \citep[OU along a branch with a normal jump, denoted JOU][]{Bartoszek:2014ca,Bokma:2002hk}.
To the best of our knowledge, our implementation handles the widest class of BM-- and OU--based models 
on the widest set of phylogenetic trees, including non--ultrametric and non--binary trees.

It is important to stress here one point about the presented methodology and accompanying package. 
Our aim is not to provide a complete inference framework.
Rather we provide an efficient
way to evaluate the likelihood for a phylogenetic comparative data set given a user--defined model. 
The user can then on top of our package optimize over the parameter space to find the maximum likelihood 
estimates or perform a Bayesian analysis. 
In \cite{Mitov:2018pnas}, 
we use the framework presented here to quantify the evolution of brain-body mass allometry  in mammals.
\del{Quick and accurate evaluation of the likelihood is one of 
the most challenging components in an PCM analysis. It is either very involved
(our implementation or the $3$--point structure of \pkg{Rphylopars}) or results in 
painfully slow code, i.e.
one uses the multivariate normal density formula directly. Of course, for the likelihood
calculation to be efficient (linear in terms of number of tips) we need to make some assumptions, 
and the only one we have is that the transition density along a branch is multivariate normal
with expectation linearly dependent on the ancestor and variance independent of the ancestral state.
However, we already pointed out that this can be substantially relaxed.
The price to pay for relaxations will not be in the the computational
complexity of the likelihood calculations but rather in a more involved, than now, user interface. 
In summary, we envision the likelihood calculations within our package to be used together with an 
optimization tool in order to find the parameters best explaining the observed traits.} 
\TScom{didn't we already say all this above}

The rest of the paper is organized as follows. In Section \ref{secFastPCM}, we describe in detail 
our fast computational framework for phylogenetic comparative methods.  \TScom{mention also Sec 3 here?} 
\VM{In Section \ref{secSoftware} we present the \pkg{PCMBase} \proglang{R}--package.}
Then, in  Section \ref{secSpecialIssues} we describe how one can
\del{smoothly} handle issues such as missing values, measurement error, \TS{punctuated components,  
trees with polytomies, as well as sequentially sampled data (such as fossil data) leading to non--ultrametric trees.} 
Next, in Section \ref{secpcmOU}, we discuss the \del{, now} standard Ornstein--Uhlenbeck setup and describe examples of  
model classes that are already provided within our package. Two \TS{widely used models}---the
multivariate Brownian motion and multivariate Ornstein--Uhlenbeck processes and \TS{a novel model---
a multivariate Ornstein--Uhlenbeck model with jumps} are provided.
It should be noted that even though we call the BM and OU standard PCM models, our implementation
goes beyond what can be usually found in implementations: 
\TS{First, we allow for non-ultrametric trees}. 
\TS{Second, the} only assumption that we make on the drift matrix
(i.e. ``deterministic part'' of the SDE) is that it has to be eigendecomposable. This is in contrast to 
the assumption of this matrix being not only eigendecomposable but also non--singular (e.g. \pkg{Rphylopars}, 
\pkg{mvMORPH}, \pkg{mvSLOUCH}---but some exceptions to this are permitted). 
In Section \ref{secPostQuant} we report a technical validation test of the likelihood calculations
and Section \ref{secDiscussion} is a discussion.. 

\section[Fast phylogenetic computational framework]{Fast phylogenetic computational framework}\label{secFastPCM}
\subsection{Phylogenetic notation}
We assume that we are given a rooted phylogenetic tree $\mathbb{T}$ representing the ancestral relationship between $N$ 
species associated with the tips of the tree (fig. \ref{figTreeNotation}). 
We denote the tips of the tree by the numbers $1,\ldots,N$,
the internal nodes by the numbers $N+1,\ldots,M-1$ 
(where $M$ is the total number of nodes in the tree)
and the root-node by $0$.
For any internal node $j$, we denote by $Desc(j)$ the set of its direct descendants.
We denote by $\mathbb{T}_{j}$ the subtree rooted at node $j$.
We denote by $t_{j}$ the known length of the branch in the tree
leading to any tip or internal node $j$. 
By convention, we assume that time increases in the direction from the root to the tips of the tree, 
and $t_{j}$ are positive scalars.

\begin{figure}[t]
\begin{center}
\includegraphics[width=0.5\textwidth]{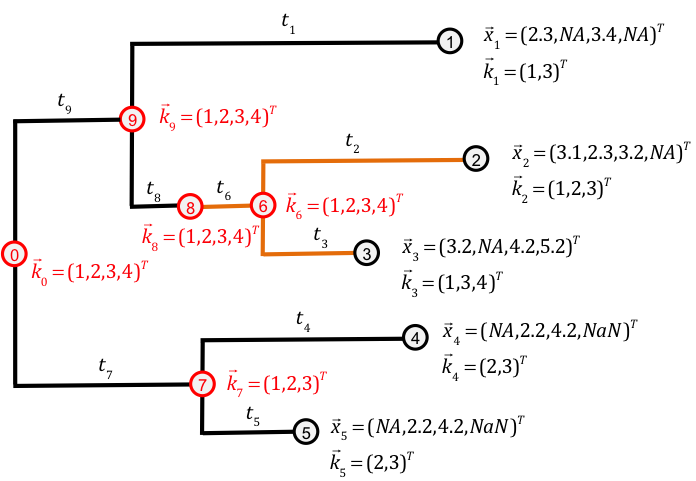}
\end{center}
\caption{A phylogenetic tree with observations at the tips. 
Numbered circles in black indicate the tips with observed trait vectors 
$\vec{x}_{1},\ldots,\vec{x}_{N=5}$. 
Missing measurements are denoted as \code{NA} (Not Available), 
while non-existing traits are denoted as \code{NaN} (Not a Number). Numbered circles in red indicate the root, 
$0$, and the internal nodes $6,\ldots,9$, for which the trait vectors are unknown. The  vectors, $\vec{k}_{i}$, 
denote the active coordinates for every node - for a tip-node these are all observed (neither \code{NA} nor \code{NaN}) 
coordinates; 
for an internal node, these are all the coordinates denoting traits that exist (are not \code{NaN}) for at least one of the 
tips descending from that node. The length of a branch leading to a tip or an internal node is known and denoted by 
$t_{i}$, $i=1,\ldots,9$. The change in branch color from black to orange at the internal node $8$ denotes the change 
to a different evolutionary regime. It is assumed that such a regime change occurs simultaneously for all traits.}
\label{figTreeNotation}
\end{figure}

The object of all phylogenetic models discussed here will be a suite of $k$ quantitative (real--valued) 
traits characterizing the $N$ species. Associated with each tip, $i$, there is a real $k$--vector, $\vec{x}_{i}$, of 
measured values for the $k$ traits. For some species, some trait measurements can be missing, reflecting two possible 
cases: 
\begin{itemize}
\item the trait exists but was not measured for that species, denoted as  \code{NA} (Not Available);
\item the trait does not exist for that species denoted as \code{NaN} (Not a Number) (fig. \ref{figTreeNotation}). 
\end{itemize}

We introduce algebraic notation that will hold for the rest of the paper. Scalars
are denoted by lower case letters, e.g. $f$, vectors are indicated by the arrow notation, e.g.
$\vec{\theta}$, while matrices are denoted
as upper case bold letters, e.g. $\mathbf{H}$. An exception to this is $\mathbf{X}_{j}$, meaning 
the set of measurements at the tips descending from an 
internal node $j$ of the tree. 

\VM{\subsection{Phylogenetic models of continuous trait evolution}}
We assume that the  trait values measured at the tips of the tree result from a \VM{continuous time continuous state--space} Markovian process evolving on 
top of the branching pattern in the tree.
By this we mean that along any given branch we have a trajectory
following the law of the process. Then, at speciation, the process
``splits'' into two processes. \VM{Both processes inherit the last
value of their parent process. \del{These daughter processes evolve
independently from their point of splitting. a}After the branching points,
there is no interaction between the processes. This entails that all the dependencies
between the values at the tips come from the time between the origin of the tree
and the most recent common ancestor for each pair of species. Exactly
how this shared time of evolution is translated into a dependency depends
on the assumed process. A widely used example of such trait process is the Ornstein--Uhlenbeck process illustrated 
in Fig. \ref{figPCMproc}.}

\begin{figure}[!ht]
\begin{center}
\includegraphics[width=0.4\textwidth,angle=270,origin=c]{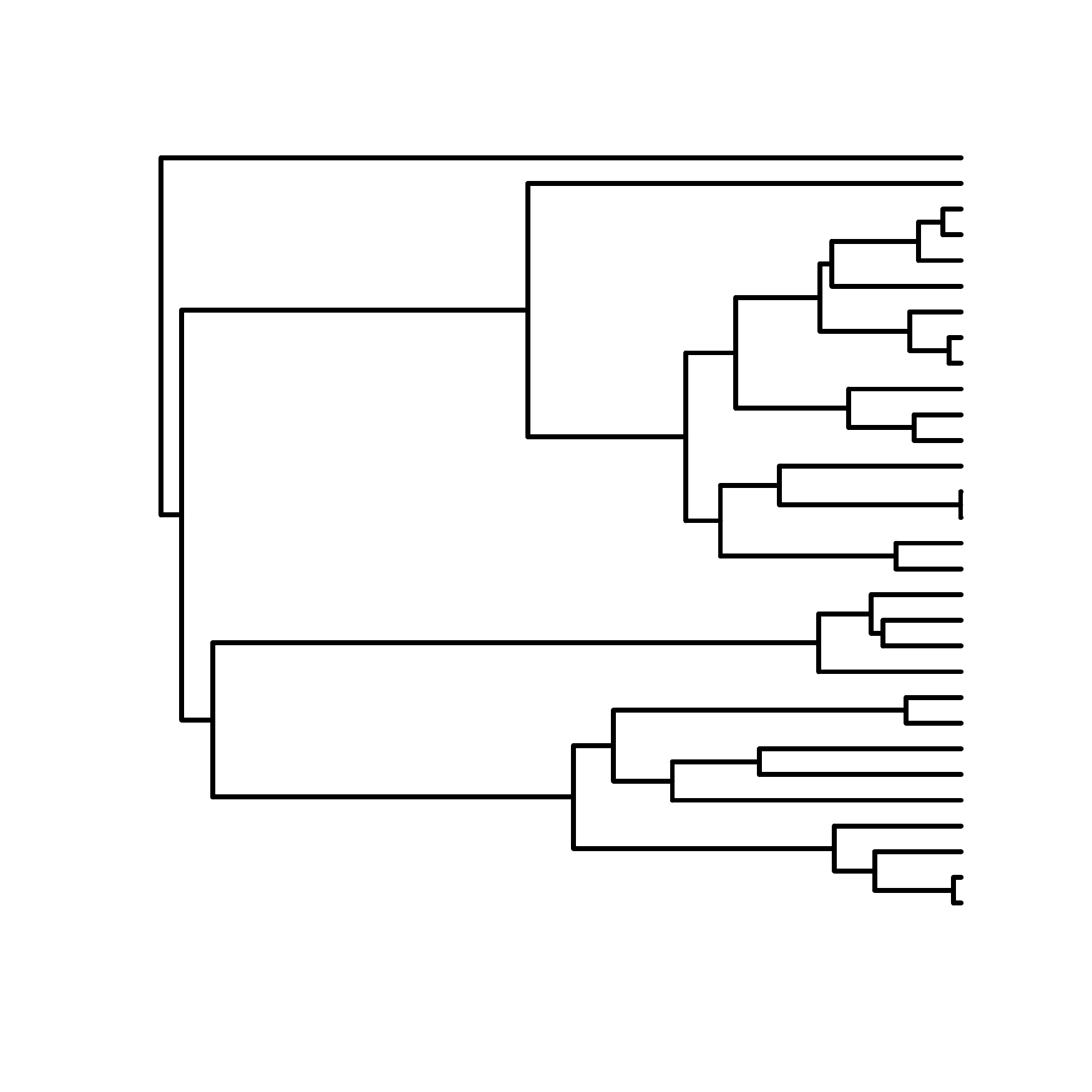}
\includegraphics[width=0.4\textwidth]{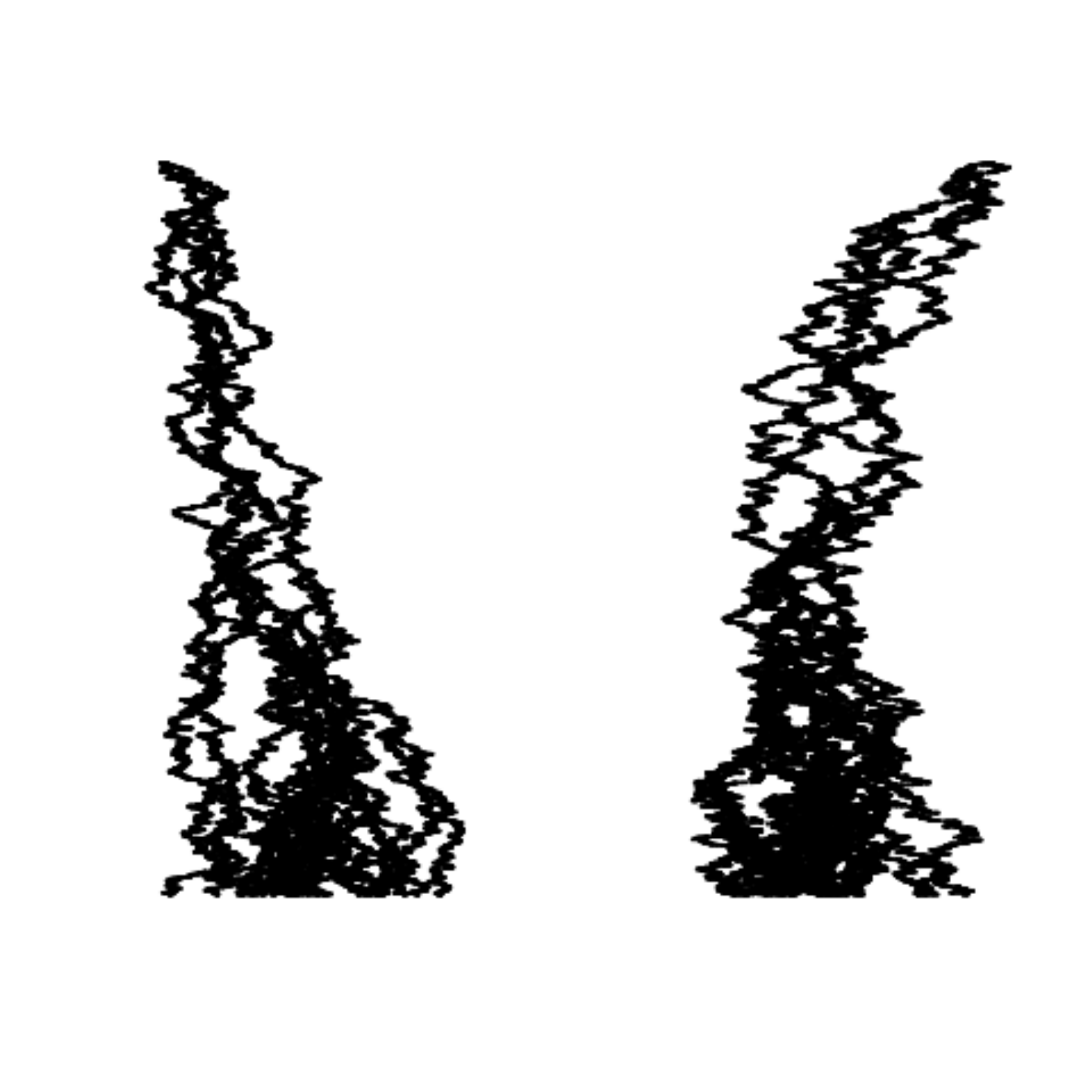}
\end{center}
\caption{
Simulation of a bivariate OU process on top of a pure birth tree with $30$ tips. 
The two traits are displayed on separate panels. 
The tree was simulated using the \pkg{TreeSim} package \citep{Stadler:2009ei,Stadler:2011ig}, its height is $3.201$.
The bivariate OU process was simulated using \pkg{mvSLOUCH} \citep{Bartoszek:2012fw}
with parameters (matrices are represented by their rows)
$\mathbf{H} = \{\{1,0.25\},\{0, 2\}\},~
\mathbf{\Sigma}_{x} = \{\{0.5,0.25\},\{0,0.5\}\},~
\vec{\theta} = (1,-1)^{T} ~\mathrm{and}~
\vec{x}_{0} = (0,0)^{T}.
$
}\label{figPCMproc}
\end{figure}


\VM{\del{
Further, in order to allow for noisy measurements or a non--phylogenetic influence on the trait, we
assume that the observable measurement presents the sum of
the random variables evolving according to the stochastic process and a random error, $\vec{e}$.
We assume that the error vector $\vec{e}$ is uncorrelated with the stochastic process and
distributed according to a multivariate normal distribution with mean $\vec{0}$ and
variance--covariance matrix $\mathbf{\Sigma}_{e}$. It is especially important
to remember about this error term when one has multiple individuals measured inside a species.
Using it allows one to incorporate intra--species variation.}}

\VM{
Such stochastic processes are used as models of continuous trait evolution at the macro-evolutionary time scale, 
that is, when the time-units are in the order of hundreds to thousands of generations. 
Further in the text, we use the term ``(trait evolutionary) model'' to denote such kind of stochastic processes. 
We now turn to describing a family of models for which we will then provide an efficient way to calculate the 
likelihood of their parameters given the tree and the trait data observed at its tips. 
\VMcom{We consider the tree as fixed, so I think, it should be part of the observed data, not a model parameter.}}

\subsection{The 
$\mathcal{G}_{LInv}$ family of models} 
\del{We depart from the assumption that the trait--values at the tips of the tree come from a realization of a
\del{normal Markovian branching} 
branching continuous state--space stochastic process
(with some assumptions on the transition density). 
}
\VM{The following definition specifies all requirements needed for a  
trait evolutionary model to be integrated within the fast computational framework:
\begin{definition}[The $\mathcal{G}_{LInv}$ family]\label{defGLinv}
We say that a trait evolutionary model belongs to the $\mathcal{G}_{LInv}$ family if it satisfies the following
\begin{enumerate}
\item after branching the traits evolve independently in \VM{\del{the two}all} descending lineages,
\item the distribution of the trait vector at time $t$, $\vec{x}(t)$,  conditional on the trait vector at time 
$s < t$, $\vec{x}(s)$, is Gaussian
with the mean and variance satisfying
\begin{enumerate}[label=(2.\alph*)]
\item $\E{\vec{x}(t) \vert \vec{x}(s)} = \vec{\omega} + \mathbf{\Phi} \vec{x}(s)$ \\ (the expectation is a 
linear function of the ancestral value), \\
\item $\var{\vec{x}(t) \vert \vec{x}(s)} = \mathbf{V}$ \\ (variance is invariant with respect to the 
ancestral value),
\end{enumerate}
for some vector $\vec{\omega}$ and matrices $\mathbf{\Phi}$, $\mathbf{V}$ which may depend on $s$ and $t$ 
but do not depend on the trait trajectory $\vec{x}(\cdot)$.
\KBcom{It is unclear in the way it was written,
it cannot depend o $\vec{x}(s)$, but what about $\vec{x}(t)$}
\end{enumerate}
\end{definition}
}

\del{ representation for
normal trait generating processes.
We assume that all model parameters are piecewise constant
over the tree [I guess we never mention though that Thm 1,2 can be readily applied for piecewise constant parameters?]. 
As this is a normal process framework}
\TScom{what comes here seems to belong to Thm. 2.1.}
\KBcom{Yes, but I thought an introductory paragraph like this would be better instead of a very long theorem.
In any case this notation is also used throughout the paper, so it can be standalone like this, I think?
We wanted with Venelin to have a place to introduce some of this notation and maybe this a good one?}
\VMcom{Is introducing Theta and peace-wise constancy of the parameters needed at this point?}\del{Let us denote the set of all parameters of our considered model as $\Theta$. \TS{We assume that all model parameters are piecewise constant over the tree. }}
\TScom{I think what comes now until Remark is very confusing. shouldn't def. 1 come here? I am fine with introducing notation before the Thm. But we seem to state results of the theorem (Eqn 1) without proper context. I would define $V,\Phi, \omega$ prior to the thm, and then state in thm the cond. probability (Eqn 1) and say how $V,\Phi, \omega$ give rise to $A,b,C,d,E,f$. As I'm a bit confused about the content and you may disagree, I did not do this re-ordering yet.  }\VMcom{Did the reordering. Please note also, that the definition of GLinv was not full which is also the case for the same def. in the PNAS paper. It is important to include in the definition that all of omega, Phi and V do not depend on x(s). Also the theorem was not stated correctly, because it was omitting to mention that A,b,C,d,E,f are also independent with respect to xj. Does the maths in thm 1 make sense now?}
\del{ pdf of a trait vector at a node in 
$\mathbb{T}$ conditioned on the trait vector at its parent node 
  All of these are composite parameters, i.e. functions of elements
of $\Theta$ and $t_{i}$---the length of the branch leading to node $i$.
We also introduce here some notation that will be used throughout the paper and will
be also useful for the user of our package.
We may represent a non--degenerate (i.e. the covariance matrix is non--singular), appropriately conditioned on 
$\vec{x}_{j}$ and $t_{i}$,  normal density function as
Such a representation for the pdf is equivalent to saying that the trait
process is multivariate normal.
The ancestral state enters the conditional mean linearly and does not
affect the conditional variance.
For a pair of normal random variables the conditional expectation
of one on the other is linear in terms of the other and the conditional variance is independent
of the value of the conditioned on random variable. We summarize this in the following theorem.
}

\VM{
Later, in section \ref{secpcmOU}, we show that the  $\mathcal{G}_{LInv}$ family contains many \TS{well--known} 
\TS{contemporary} models
such as BM, multivariate OU (where all traits are OU or some are BMs), BM or OU with normally distributed jumps. 
Now we derive an important property of the $\mathcal{G}_{LInv}$ family playing a key role for the fast 
likelihood calculation:}

\VMcom{Showing the equivalence between the GLinv and the quadratic polynomial representation is not needed at this point, 
but is needed later to clarify the scope of the computational framework. 
Hence, I've split the theorem 1 into theorem 1 and theorem 3 (see next subsection). 
In this way, the next subsection can be written a bit more explicit (without length concerns) -- 
the questions about the symmetry and the negative definiteness of A should indeed be mentioned.
 Otherwise, it is not clear, why do we want A to be symmetric negative -- for example the integration in 
 thm 2 works also for asymmetric A.  Hence, remark 1 is indeed useful (i think). }
\KBcom{The way remark 1 is written later on is wrong. I have removed the part with the error. Since A = -2V
then it is obvious that it has to be symmetric neg-def. As we introduce the model through omega, Phi, V 
then we induce certain properties on AbCdEf. In the previous version we just took AbCdEf so we had to justify the
choice of constraints. Here, these are induced. Therefore, I do not think that the remark is required. I have left
it in a correct version.
}
\begin{theorem}\label{thmpcmNorm}
Let $\mathcal{M}$ be a trait model from the $\mathcal{G}_{LInv}$-family. Let $i$ be a tip or internal node and $j$ 
be its parent node in a tree $\mathbb{T}$, and 
let $\vec{x}_{i}\in \mathbb{R}^{k_{i}}$, $\vec{x}_{j}\in \mathbb{R}^{k_{j}}$ ($k_{i},k_{j}\in\mathbb{Z}^{+}$) 
be the trait--vectors at the nodes $i$ and $j$ under a realization of $\mathcal{M}$ on $\mathbb{T}$. 
Let $\vec{\omega}_{i}$, $\mathbf{\Phi}_{i}$ and $\mathbf{V}_{i}$ denote the terms $\vec{\omega}$, $\mathbf{\Phi}$ 
and $\mathbf{V}$ from definition \ref{defGLinv} specific for node $i$. 
\KBcom{Below is stating the obvious as it is Gaussian}
Then, 
\del{
$\vec{x}_{i}$ has non--zero support on the whole of $\mathbb{R}^{k_{i}}$, and}
the probability density function (pdf) of $\vec{x}_{i}$ 
conditioned on $\vec{x}_{j}$ can be expressed as the following exponential of a quadratic polynomial 
\be\label{pdfXiXj}
pdf(\vec{x}_{i}|\vec{x}_{j})=\exp\bigg[\vec{x}_{i}^T \mathbf{A}_{i} \vec{x}_{i} + \vec{x}_{i}^T \vec{b}_{i}
+ \vec{x}_{j}^T \mathbf{C}_{i} \vec{x}_{j} + \vec{x}_{j}^T \vec{d}_{i} + \vec{x}_{j}^T \mathbf{E}_{i} \vec{x}_{i} 
+ f_{i}\bigg],
\ee
\noindent where $\mathbf{A}_{i}$ is a symmetric negative-definite matrix, and all of the terms 
$\mathbf{A}_{i}$, $\vec{b}_{i}$, $\mathbf{C}_{i}$, $\vec{d}_{i}$, $\mathbf{E}_{i}$, $f_{i}$ are 
constants with respect to $\vec{x}_{j}$, specified by the equations:
\be
\begin{array}{rcl}
\mathbf{A}_{i} & = & -\frac{1}{2}\mathbf{V}_{i}^{-1} \in \mathbb{R}^{k_{i}\times k_{i}} \\
\vec{b}_{i} & = & \mathbf{V}_{i}^{-1}\vec{\omega}_{i} \in \mathbb{R}^{k_{i}} \\
\mathbf{C}_{i} & = &-\frac{1}{2}\mathbf{\Phi}_{i}^{T}\mathbf{V}_{i}^{-1}\mathbf{\Phi}_{i} \in \mathbb{R}^{k_{j}\times k_{j}}\\
\vec{d}_{i} & = & -\mathbf{\Phi}_{i}^{T}\mathbf{V}_{i}^{-1}\vec{\omega}_{i} \in \mathbb{R}^{k_{j}}\\
\mathbf{E}_{i} & = & \mathbf{\Phi}_{i}^{T}\mathbf{V}_{i}^{-1} \in \mathbb{R}^{k_{j}\times k_{i}} \\
f_{i} & = & -\frac{1}{2}\vec{\omega}_{i}^{T}\mathbf{V}_{i}^{-1}\vec{\omega}_{i} -\frac{k_{i}}{2}\log\left(2\pi\right) - \frac{1}{2}\log \vert \mathbf{V}_{i}\vert \in \mathbb{R}.
\end{array}\label{eqAbCdEfVoP}
\ee
\begin{proof}
Substituting $\vec{\omega}_{i} + \mathbf{\Phi}_{i} \vec{x}_{j}$ and $\mathbf{V}_{i}$ for the mean and 
variance in the formula for the pdf of a multivariate Gaussian distribution, we obtain:
\be\label{eqmvnorm}
pdf(\vec{x}_{i}|\vec{x}_{j}) = \exp\bigg[-\frac{1}{2}\left(\vec{x}_{i} - \left(\vec{\omega}_{i} + \mathbf{\Phi}_{i}\vec{x}_{j} \right)\right)^{T}\mathbf{V}_{i}^{-1}\left(\vec{x}_{i} - \left(\vec{\omega}_{i} + \mathbf{\Phi}_{i}\vec{x}_{j} \right)\right)
-\frac{k_{i}}{2}\log\left(2\pi\right) - \frac{1}{2}\log \vert \mathbf{V}_{i}\vert \bigg]
\ee
\noindent By expanding and reordering the terms in parentheses, Eq. \eqref{eqmvnorm} can be rewritten as
\be\label{eqmvnorm2ndForm}
\begin{array}{llll}
pdf(\vec{x}_{i}|\vec{x}_{j}) & = & \exp\bigg[ & \vec{x}_{i}^{T}\big(-\frac{1}{2}\mathbf{V}_{i}^{-1}\big)\vec{x}_{i}+\\
 & & & \vec{x}_{i}^{T}\big(\mathbf{V}_{i}^{-1}\vec{\omega}_{i}\big)+\\
 & & & \vec{x}_{j}^{T}\big(-\frac{1}{2}\mathbf{\Phi}_{i}^{T}\mathbf{V}_{i}^{-1}\mathbf{\Phi}_{i}\big)\vec{x}_{j}+\\
 & & & \vec{x}_{j}^{T}\big(-\mathbf{\Phi}_{i}^{T}\mathbf{V}_{i}^{-1}\vec{\omega}_{i}\big)+\\
 & & & \vec{x}_{j}^{T}\big(\mathbf{\Phi}_{i}^{T}\mathbf{V}_{i}^{-1}\big)\vec{x}_{i}+\\
 & & & \big(-\frac{1}{2}\vec{\omega}_{i}^{T}\mathbf{V}_{i}^{-1}\vec{\omega}_{i}
-\frac{k_{i}}{2}\log\left(2\pi\right) - \frac{1}{2}\log \vert \mathbf{V}_{i}\vert\big)
\bigg].
\end{array}
\ee
We can see the correspondence with the quadratic forms 
$\vec{x}_{i}^T(\ldots)\vec{x}_{i}$, $\vec{x}_{j}^{T}(\ldots)\vec{x}_{j}$ 
and the other terms in Eq. \eqref{pdfXiXj}. Equation \eqref{eqAbCdEfVoP} 
follows immediately. Furthermore, $\mathbf{A}_{i}$ is a symmetric negative--definite matrix, 
because $\mathbf{V}_{i}$ is a symmetric positive--definite matrix as it is a
\del{by the definition of a} 
variance--covariance matrix. \del{of a multivariate Gaussian distribution.}
Finally, all of the terms $\mathbf{A}_{i}$, $\vec{b}_{i}$, $\mathbf{C}_{i}$, $\vec{d}_{i}$, $\mathbf{E}_{i}$, $f_{i}$ 
are constant with respect to $\vec{x}_{j}$, because they are functions of $\vec{\omega}_{i}$, $\mathbf{\Phi}_{i}$ 
and $\mathbf{V}_{i}$ which are constants with respect to $\vec{x}_{j}$ by Def. \ref{defGLinv}.
\end{proof}
\end{theorem}

\subsection{Calculating the likelihood of 
$\mathcal{G}_{LInv}$-models}\label{subsecmvFast}
Let $\mathcal{M}$ be a trait evolutionary model realized on a tree $\mathbb{T}$ and $\Theta$ denotes the 
parameters of $\mathcal{M}$. The likelihood of $\mathcal{M}$ for given trait data $\mathbf{X}$ associated 
with the tips of $\mathbb{T}$ is defined as the function $\ell(\Theta)=pdf(\mathbf{X}\vert \mathbb{T},\Theta)$. 
The representation of Eq. \eqref{pdfXiXj} allows for linear (in terms of \VM{the number of tips, $N$}) calculation
of the likelihood of \VM{any trait model in the 
$\mathcal{G}_{LInv}$-family, given a phylogeny and measured 
data at its tips.} \TScom{linear follows from next theorem, right? we should say that.} 
This follows from the next theorem.

\begin{theorem}\label{thmMvFast}
Let $\mathcal{M}$ be a trait evolutionary model from the $\mathcal{G}_{LInv}$-family and $\mathbb{T}$ 
be a phylogenetic tree. Let $\Theta$ be the parameters of $\mathcal{M}$. For the root (0) or any internal 
node $j$ in $\mathbb{T}$,
there exists a $k_{j}\times k_{j}$ matrix $\mathbf{L}_{j}$, a $k_{j}$--vector $\vec{m}_{j}$
and a scalar $r_{j}$, such that the likelihood of $\mathcal{M}$ for the data 
\TScom{for [likelihood is function OF parameters, not OF data]} \del{of} $\mathbf{X}_{j}$, 
conditioned on $\vec{x}_{j} \in \mathbb{R}^{k_{j}}$ and \VM{
\del{
$\mathbb{T}_{j}$ (the subtree with node $j$ as its root)}} 
\VM{
$\mathbb{T}$} is expressed as: \VMcom{Previously, 
we wrote $\mathbb{T}_{j}$ 
but it should be the whole tree, because omega, Phi and V can depend on the whole tree, 
i.e. the branch-lengths before the node j. I've also updated all equations below.} 

\be\label{pdfXXjXj}
pdf(\mathbf{X}_{j}|\vec{x}_{j},\mathbb{T},\Theta)=\exp\left(\vec{x}_{j}^T\mathbf{L}_{j}\vec{x}_{j}+
\vec{x}_{j}^T\vec{m}_{j}+r_{j}\right).
\ee
The parameters $\mathbf{L}_{j}$, $\vec{m}_{j}$, $r_{j}$ 
are functions of $\Theta$, the observed data $\mathbf{X}_{j}$, and 
\VM{the tree 
$\mathbb{T}$}\VMcom{I would remove this clause, because the topology and the branch 
length are not terms in the mentioned equations:\del{, both in terms of topology and branch lengths}}, 
\TS{namely, equations \ref{pdfLmn_alltips}, \ref{pdfLmn_tips}, and \ref{pdfLmn_non-tips}.}

\end{theorem}

\begin{proof}
\VM{Following condition 1 of definition \ref{defGLinv} we can factorize} the conditional likelihood at any
internal or root node $j$. Splitting $Desc(j)$, i.e. the set of nodes descending from node $j$,
into tips and non--tips,
denoted as $Desc(j)\cap \{1,...,N\}$ and $Desc(j)\setminus \{1,...,N\}$, we can write:

\be\label{pdfXXjXj_fact}
\begin{array}{rcl}
pdf(\mathbf{X}_{j}|\vec{x}_{j},\mathbb{T},\Theta) & = & \left(\prod\limits_{i \in Desc(j)\cap\{1,...,N\}}
pdf(\vec{x}_{i}|\vec{x}_{j},\mathbb{T},\Theta) \right)
\times \\
& & \left(\prod\limits_{i \in Desc(j)\setminus\{1,...,N\}}\bigintss_{\mathbb{R}^{k}_{j}}pdf(\vec{x}_{i}|\vec{x}_{j},\mathbb{T},\Theta)
\times pdf(\mathbf{X}_{i}|\vec{x}_{i},\mathbb{T},\Theta)\ud\vec{x}_{i}\right).
\end{array}
\ee

We first prove the Theorem for nodes where all descendants are tips.
If all descendants of $j$ are tips (e.g. nodes 6 and 7 on Fig. \ref{figTreeNotation}),
then, according to Eq. \eqref{pdfXiXj}

$$
\begin{array}{rcl}
pdf(\mathbf{X}_{j}|\vec{x}_{j},\mathbb{T},\Theta) & = & \prod\limits_{i \in Desc(j)}pdf(\vec{x}_{i}|\vec{x}_{j},\mathbb{T},\Theta)\\
& = & \exp \left( \sum\limits_{i \in Desc(j)} \vec{x}_{i}^T \mathbf{A}_{i} \vec{x}_{i} + \vec{x}_{i}^T \vec{b}_{i} +
\vec{x}_{j}^T \mathbf{C}_{i} \vec{x}_{j} + \vec{x}_{j}^T \vec{d}_{i} + \vec{x}_{j}^T\mathbf{E}_{i} \vec{x}_{i} + f_{i} \right),
\end{array}
$$
resulting in

\be\label{pdfXXjXj_alltips}
pdf(\mathbf{X}_{j}|\vec{x}_{j},\mathbb{T},\Theta) =  \exp\left(\vec{x}_{j}^T (\sum\limits_{i \in Desc(j)} \mathbf{C}_{i}) \vec{x}_{j} +
\vec{x}_{j}^{T}(\sum\limits_{i\in Desc(j)}\vec{d}_{i} + \mathbf{E}_{i} \vec{x}_{i}) +
\sum\limits_{i \in Desc(j)}\vec{x}_{i}^T \mathbf{A}_{i} \vec{x}_{i} + \vec{x}_{i}^T \vec{b}_{i} + f_{i}\right)
\ee
Then, to obtain the representation from Eq. \eqref{pdfXXjXj}, we denote:

\be\label{pdfLmn_alltips}
\begin{array}{rcl}
\mathbf{L}_{j}& = &\sum\limits_{i \in Desc(j)} \mathbf{C}_{i}\\
\vec{m}_{j}& = &\sum\limits_{i\in Desc(j)}\vec{d}_{i} + \mathbf{E}_{i} \vec{x}_{i}\\
r_{j} & = &\sum\limits_{i \in Desc(j)}\vec{x}_{i}^T \mathbf{A}_{i} \vec{x}_{i} + \vec{x}_{i}^T \vec{b}_{i} + f_{i}
\end{array}
\ee
If not all of $Desc(j)$ are tips, then, for the descendants which are tips, we define:

\be\label{pdfLmn_tips}
\begin{array}{rcl}
\mathbf{L}_{j}^{tips}& = &\sum\limits_{i \in Desc(j)\cap \{1,...,N\}} \mathbf{C}_{i}\\
\vec{m}_{j}^{tips}& = &\sum\limits_{i\in Desc(j)\cap \{1,...,N\}}\vec{d}_{i} + \mathbf{E}_{i} \vec{x}_{i}\\
r_{j}^{tips} & = &\sum\limits_{i \in Desc(j)\cap \{1,...,N\}}\vec{x}_{i}^T \mathbf{A}_{i} \vec{x}_{i}
+ \vec{x}_{i}^T \vec{b}_{i} + f_{i}
\end{array}
\ee

We perform mathematical induction to prove the Theorem for all nodes. We need to show that
\VM{Eq. \eqref{pdfXXjXj} holds for each non-tip descendant of $j$, that is,} 
for each $i\in Desc(j)\setminus \{1,...,N\}$
there exists a $k_{i}\times k_{i}$ matrix $\mathbf{L}_{i}$, a $k_{i}$--vector $\vec{m}_{i}$ and a
scalar $r_{i}$ such that
$pdf(\mathbf{X}_{i}\vert \vec{x}_{i},\mathbb{T},\Theta)=\exp(\vec{x}_{i}^T\mathbf{L}_{i}\vec{x}_{i}+\vec{x}_{i}^T\vec{m}_{i}+r_{i})$.
We proved the induction base case, namely, we proved above \VM{that} the 
Eq. \eqref{pdfXXjXj} holds for all nodes which have only tip--descendants. 
Then, the induction hypothesis is that for an internal node $j$, the statement of the 
theorem has been proven for all \VM{
\del{subtrees 
$\mathbb{T}_{i}$, such that}} $i \in Desc(j)$. 
Now in the inductive step using Eq. \eqref{pdfXiXj} and the induction hypothesis, 
we can write the integral in Eq. \eqref{pdfXXjXj_fact} as

\begin{multline*}
\int_{\mathbb{R}^{k_{i}}}pdf(\vec{x}_{i}|\vec{x}_{j},\mathbb{T},\Theta) \times pdf(\mathbf{X}_{i}|\vec{x}_{i},\mathbb{T},\Theta)\ud\vec{x}_{i}
\\ = \int_{\mathbb{R}^{k_{i}}}\exp \left(\vec{x}_{i}^T \mathbf{A}_{i} \vec{x}_{i} + \vec{x}_{i}^T \vec{b}_{i}
+ \vec{x}_{j}^T \mathbf{C}_{i} \vec{x}_{j} + \vec{x}_{j}^T \vec{d}_{i} + \vec{x}_{j}^T \mathbf{E}_{i} \vec{x}_{i}
+ f_{i} + \vec{x}_{i}^T\mathbf{L}_{i}\vec{x}_{i}+\vec{x}_{i}^T\vec{m}_{i}+r_{i}\right)\ud\vec{x}_{i}\\
= \exp \left(\vec{x}_{j}^T \mathbf{C}_{i} \vec{x}_{j} + \vec{x}_{j}^T \vec{d}_{i} + f_{i} + r_{i}\right)
\times
\boxed{\int_{\mathbb{R}^{k_{i}}}\exp \left(\vec{x}_{i}^{T} (\mathbf{A}_{i}+\mathbf{L}_{i}) \vec{x}_{i} + \vec{x}_{i}^T (\vec{b}_{i} + \vec{m}_{i}+
\mathbf{E}_{i}^{T}\vec{x}_{j}) \right)\ud\vec{x}_{i}} \\
\stackrel{\bigstar}{=} \exp \left(\vec{x}_{j}^T \mathbf{C}_{i} \vec{x}_{j} + \vec{x}_{j}^T \vec{d}_{i} + f_{i} + r_{i}\right)
\left(\sqrt{2\pi} \right)^{k_{i}}\boxed{ \left(\sqrt{\vert (-2)\left(\mathbf{A}_{i} + \mathbf{L}_{i} \right) \vert} \right)^{-1}}
\\ \times
\exp\left(-(1/4)
\left(\vec{b}_{i}+\vec{m}_{i}+\mathbf{E}_{i}^{T}\vec{x}_{j}\right)^{T}\left(\mathbf{A}_{i} + \mathbf{L}_{i}\right)^{-1}
\left(\vec{b}_{i}+\vec{m}_{i}+\mathbf{E}_{i}^{T}\vec{x}_{j}\right)
\right)
\end{multline*}

\begin{multline*}
= \exp \left(\vec{x}_{j}^T \mathbf{C}_{i} \vec{x}_{j} + \vec{x}_{j}^T \vec{d}_{i} + f_{i} + r_{i}\right)
\left(\sqrt{2\pi} \right)^{k_{i}} \left(\sqrt{\vert (-2)\left(\mathbf{A}_{i} + \mathbf{L}_{i} \right) \vert} \right)^{-1}
\\ \times
\exp\left(
-(1/4)\left(\vec{b}_{i}+\vec{m}_{i}\right)^{T}\left(\mathbf{A}_{i} + \mathbf{L}_{i}\right)^{-1}
\left(\vec{b}_{i}+\vec{m}_{i}\right)
-(1/2) \vec{x}_{j}^{T}\mathbf{E}_{i}\left(\mathbf{A}_{i} + \mathbf{L}_{i}\right)^{-1}
\left(\vec{b}_{i}+\vec{m}_{i}\right)
\right. \\ \left.
-(1/4)
\vec{x}_{j}^{T}\mathbf{E}_{i}\left(\mathbf{A}_{i} + \mathbf{L}_{i}\right)^{-1}\mathbf{E}_{i}^{T}\vec{x}_{j}
\right)
\\
=\exp \left(\vec{x}_{j}^T \left(\mathbf{C}_{i}-(1/4) \mathbf{E}_{i}\left(\mathbf{A}_{i} + \mathbf{L}_{i}\right)^{-1}\mathbf{E}_{i}^{T} \right) \vec{x}_{j}
+ \vec{x}_{j}^T \left(\vec{d}_{i} -(1/2) \mathbf{E}_{i}\left(\mathbf{A}_{i} + \mathbf{L}_{i}\right)^{-1}
\left(\vec{b}_{i}+\vec{m}_{i}\right)\right)
\right. \\ \left.
+ f_{i} + r_{i} +
(k_{i}/2)\log(2\pi) -(1/2) \log(\vert (-2)\left(\mathbf{A}_{i} + \mathbf{L}_{i} \right) \vert)
-(1/4)\left(\vec{b}_{i}+\vec{m}_{i}\right)^{T}\left(\mathbf{A}_{i} + \mathbf{L}_{i}\right)^{-1}\left(\vec{b}_{i}+\vec{m}_{i}\right)
\right)
\end{multline*}
We can then see that for a non--tip node we can define

\be\label{pdfLmn_non-tips}
\begin{array}{rcl}
\mathbf{L}_{j}^{non-tips}& = &\sum\limits_{i \in Desc(j)\setminus \{1,...,N\}}
\left(
\mathbf{C}_{i} - (1/4) \mathbf{E}_{i}\left(\mathbf{A}_{i} + \mathbf{L}_{i}\right)^{-1}\mathbf{E}_{i}^{T} \right) \\
\vec{m}_{j}^{non-tips}& = &\sum\limits_{i\in Desc(j)\setminus \{1,...,N\}}
\left( \vec{d}_{i}  -(1/2) \mathbf{E}_{i}\left(\mathbf{A}_{i} + \mathbf{L}_{i}\right)^{-1}
\left(\vec{b}_{i}+\vec{m}_{i}\right) \right)
\\
r_{j}^{non-tips} & = &\sum\limits_{i \in Desc(j)\setminus \{1,...,N\}}\left(  f_{i} + r_{i}
+ (k_{i}/2)\log(2\pi) -(1/2) \log(\vert (-2)\left(\mathbf{A}_{i} + \mathbf{L}_{i} \right) \vert)
\right. \\ && \left.
-(1/4)\left(\vec{b}_{i}+\vec{m}_{i}\right)^{T}\left(\mathbf{A}_{i} + \mathbf{L}_{i}\right)^{-1}
\left(\vec{b}_{i}+\vec{m}_{i}\right) \right).
\end{array}
\ee
The representation of $\mathbf{L}_{j}^{non-tips}$, $\vec{m}_{j}^{non-tips}$ and $r_{j}^{non-tips}$
in  Eq. \eqref{pdfLmn_non-tips} immediately entails the existence of the
$\mathbf{L}_{j}$, $\vec{m}_{j}$ and $r_{j}$ elements in Eq. \eqref{pdfXXjXj} for internal or root
nodes $j$, hence we obtain the claimed polynomial form in the inductive step and in consequence the theorem.
\end{proof}

The inductive proof of Thm. \ref{thmMvFast} defines a pruning--wise procedure for calculating 
$\mathbf{L}_{0}$, $\vec{m}_{0}$ and $r_{0}$ (we remind that $0$ stands for the root of the tree). 
In order to calculate the likelihood of the tree conditioned on $\vec{x}_{0}$, 
we use Thm \ref{thmMvFast} with $j$ being the root node.
In order to be able to calculate the \TS{full} likelihood, it
now only remains to specify how to deal with the \TS{unknown} trait value at the root of the tree, \del{i.e.} $\vec{x}_{0}$,
\TS{i.e.} the ancestral state. This is an implementation detail up to the user.
Our implementation of the various models provided (sections \ref{secSoftware} and \ref{secpcmOU}) 
with the \pkg{PCMbase} package  allow
for maximizing the polynomial with respect to $\vec{x}_{0}$ or for treating it as a free parameter
(like the elements of the parameter set $\Theta$) that the user provides.

\TScom{why not put this definition prior to Thm 1, and then proof thm 1,2? It is confusing here.}\VMcom{Done.}

\subsection{Scope of the framework}\label{subsecFrameworkScope}
We now investigate if there are other trait evolutionary models, beyond the $\mathcal{G}_{LInv}$--family, 
for which the likelihood can be calculated using the same recursive formulae, 
Eqs. \eqref{pdfLmn_alltips}, \eqref{pdfLmn_tips}, and \eqref{pdfLmn_non-tips}
First, since we calculate the likelihood in a recursive pruning fashion, we assume that evolution 
is independent across branches, meaning condition $1$ is a necessary condition.
In Theorem \ref{thmpcmNormInverse}, 
we prove that the condition $2$ in Def. \ref{defGLinv} is also a necessary condition.
In other words, we show that if the likelihood can be calculated via  recursion based on
Eqs. \eqref{eqAbCdEfVoP}, \eqref{pdfXXjXj_alltips}, then the model is in the  $\mathcal{G}_{LInv}$--family.

\begin{theorem}\label{thmpcmNormInverse}
Let $\mathcal{M}$ be a trait model satisfying condition $1$ of Def. \ref{defGLinv} and realized on a 
tree $\mathbb{T}$. If for every parent--child pair of nodes $<j,i>$ in $\mathbb{T}$,
the trait--vector $\vec{x}_{i}\in\mathbb{R}^{k_{i}}$ ($k_{i}\in\mathbb{Z}^{+}$) 
has non--zero support on the whole of $\mathbb{R}^{k_{i}}$ and there exist a symmetric 
negative--definite matrix $\mathbf{A}_{i}\in \mathbb{R}^{k_{i} \times k_{i}}$ and 
components 
$\vec{b}_{i}\in \mathbb{R}^{k_{i}}$, $\mathbf{C}_{i}\in \mathbb{R}^{k_{j}\times k_{j}}$, 
$\vec{d}_{i}\in \mathbb{R}^{k_{j}}$, $\mathbf{E}_{i}\in \mathbb{R}^{k_{j}\times k_{i}}$, $f_{i} \in \mathbb{R}$, 
such that, for any vector of values at the parent node, $\vec{x}_{j}\in\mathbb{R}^{k_{j}}$ ($k_{j}\in\mathbb{Z}^{+}$), 
the pdf of $\vec{x}_{i}$ conditional on $\vec{x}_{j}$ can be expressed by Eq. \eqref{pdfXiXj}, then $\mathcal{M}$ 
belongs to the ${G}_{LInv}$--family and the terms $\vec{\omega}_{i}$, $\mathbf{\Phi}_{i}$ and $\mathbf{V}_{i}$ 
denoting the terms $\vec{\omega}$, $\mathbf{\Phi}$ and $\mathbf{V}$ from Def. \ref{defGLinv} specific 
for node $i$ satisfy Eq. \eqref{eqAbCdEfVoP}. 
\begin{proof}
\VMcom{I've written the proof more explicit and fixed another typo the definition of fi in terms of A and b.} 
We rearrange the terms on the right--hand side of Eq. \eqref{pdfXiXj} as follows 
\be
\begin{array}{l}\label{pdfXiXjRearranged}
pdf(\vec{x}_{i}|\vec{x}_{j}) =  \exp\bigg[\vec{x}_{i}^T \mathbf{A}_{i} \vec{x}_{i} -2 \vec{x}_{i}^T \mathbf{A}_{i}\left((-\frac{1}{2}\mathbf{A}_{i}^{-1})\left(\vec{b}_{i}+  \mathbf{E}_{i}^{T}\vec{x}_{j}\right)\right)
+ \left(\vec{x}_{j}^T \mathbf{C}_{i} \vec{x}_{j} + \vec{x}_{j}^T \vec{d}_{i}  + f_{i}\right)\bigg]
\\ = 
\exp\bigg[\left(\vec{x}_{i} +\frac{1}{2}\mathbf{A}_{i}^{-1}\left(\vec{b}_{i}+  \mathbf{E}_{i}^{T}\vec{x}_{j}\right) \right)^{T} \mathbf{A}_{i} \left(\vec{x}_{i} +\frac{1}{2}\mathbf{A}_{i}^{-1}\left(\vec{b}_{i}+  \mathbf{E}_{i}^{T}\vec{x}_{j}\right) \right)
- \frac{1}{4}\left(\vec{b}_{i}+  \mathbf{E}_{i}^{T}\vec{x}_{j}\right)^{T} \mathbf{A}_{i}^{-1}\left(\vec{b}_{i}+  \mathbf{E}_{i}^{T}\vec{x}_{j}\right)\\
+ \left(\vec{x}_{j}^T \mathbf{C}_{i} \vec{x}_{j} + \vec{x}_{j}^T \vec{d}_{i}  + f_{i}\right)\bigg].
\end{array}
\ee
As the above is by definition a density on $\mathbb{R}^{k_{i}}$, integrating over $\vec{x}_{i}$ equals $1$.
Hence, after taking all constants with respect to $\vec{x}_{i}$ out of the integral and 
multiplying/dividing the integral by the constant $(\sqrt{\vert 2\pi (-2) \mathbf{A}_{i} \vert})^{-1}$, we obtain:

\begin{equation}\label{eqIntegralPdfXiXj}
\begin{array}{lcl}
1 & = & 
\underbrace{\bigintsss\limits_{\mathbb{R}^{k_{i}}}
\frac{1}{\sqrt{\vert 2\pi (-2) \mathbf{A}_{i} \vert}}\exp\bigg[-\frac{1}{2}\left(\vec{x}_{i} +\frac{1}{2}\mathbf{A}_{i}^{-1}\left(\vec{b}_{i}+  \mathbf{E}_{i}^{T}\vec{x}_{j}\right) \right)^{T}\left(-2 \mathbf{A}_{i} \right) \left(\vec{x}_{i} +\frac{1}{2}\mathbf{A}_{i}^{-1}\left(\vec{b}_{i}+  \mathbf{E}_{i}^{T}\vec{x}_{j}\right) \right)
\bigg]\ud \vec{x}_{i}}_{=1}
\\
& &\times \sqrt{\vert 2\pi (-2) \mathbf{A}_{i} \vert} \times \exp\bigg[- \frac{1}{4}\left(\vec{b}_{i}+  \mathbf{E}_{i}^{T}\vec{x}_{j}\right)^{T} \mathbf{A}_{i}^{-1}\left(\vec{b}_{i}+  \mathbf{E}_{i}^{T}\vec{x}_{j}\right)
+ \left(\vec{x}_{j}^T \mathbf{C}_{i} \vec{x}_{j} + \vec{x}_{j}^T \vec{d}_{i}  + f_{i}\right)\bigg]
\\
 & = &
\exp\bigg[\frac{k_{i}}{2}\log(2\pi) + \frac{1}{2}\log\vert (-2)\mathbf{A}_{i}\vert
\bigg] \times \exp\bigg[- \frac{1}{4}\left(\vec{b}_{i}+  \mathbf{E}_{i}^{T}\vec{x}_{j}\right)^{T} \mathbf{A}_{i}^{-1}\left(\vec{b}_{i}+  \mathbf{E}_{i}^{T}\vec{x}_{j}\right)
+ \left(\vec{x}_{j}^T \mathbf{C}_{i} \vec{x}_{j} + \vec{x}_{j}^T \vec{d}_{i}  + f_{i}\right)\bigg]
\\
 & = &
\exp\bigg[
\vec{x}_{j}^{T}\left( \mathbf{C}_{i}  - \frac{1}{4}\mathbf{E}_{i}\mathbf{A}_{i}^{-1}\mathbf{E}_{i}^{T}\right)\vec{x}_{j}
+\vec{x}_{j}^T \left(\vec{d}_{i}- \frac{1}{2}\mathbf{E}_{i}\mathbf{A}_{i}^{-1}\vec{b}_{i}\right)
+f_{i}+\frac{k_{i}}{2}\log(2\pi) + \frac{1}{2}\log\big(\vert (-2)\mathbf{A}_{i}\vert\big) - \frac{1}{4}\vec{b}_{i}^{T}\mathbf{A}_{i}^{-1}\vec{b}_{i}
\bigg].
\end{array}
\end{equation}
\noindent When calculating the integral in Eq. \eqref{eqIntegralPdfXiXj} above, we have used the fact that the 
matrix $(-2)\mathbf{A}_{i}$ is a symmetric positive--definite  matrix as it is the \del{double} negative of the symmetric 
negative--definite matrix $2\mathbf{A}_{i}$. Hence, the so constructed function below the integral 
in Eq. \eqref{eqIntegralPdfXiXj} is a $k_{i}$-variate Gaussian pdf with mean vector 
$-\frac{1}{2}\mathbf{A}_{i}^{-1}\left(\vec{b}_{i}+  \mathbf{E}_{i}^{T}\vec{x}_{j}\right)$ 
and variance--covariance matrix $(-2)\mathbf{A}_{i}$. 

By definition, $\mathbf{A}_{i}$, $\vec{b}_{i}$, $\mathbf{C}_{i}$,  $\vec{d}_{i}$, $\mathbf{E}_{i}$, $f_{i}$ 
are constant with respect to $\vec{x}_{j}$. Therefore, Eq. \ref{eqIntegralPdfXiXj} has to hold for all $\vec{x}_{j}$. 
This implies the relationships:
\be\label{eqpdfXiXjConstr}
\begin{array}{rcl}
\mathbf{C}_{i} & = & \mathbf{E}_{i}\mathbf{A}_{i}^{-T}\mathbf{E}_{i}^{T}, \\
\vec{d}_{i} & = & 2\mathbf{E}_{i}\mathbf{A}_{i}^{-1}\vec{b}_{i}, \\
f_{i} & = & \frac{1}{4} b_{i}^{T}\mathbf{A}_{i}^{-1}b_{i} - \frac{k_{i}}{2}\log(2\pi)  - \frac{1}{2}\log\big( \vert (-2)\mathbf{A}_{i} \vert\big).
\end{array}
\ee

Next, we define $\mathbf{V}_{i}:=(-\frac{1}{2})\mathbf{A}_{i}^{-1}$, $\vec{\omega}_{i}:=(-\frac{1}{2})\mathbf{A}_{i}^{-1}\vec{b}_{i}$
and $\mathbf{\Phi}_{i}:=(-\frac{1}{2})\mathbf{A}_{i}^{-1}\mathbf{E}_{i}^{T}$. Since $\mathbf{A}_{i}$ 
is symmetric negative--definite, $\mathbf{V}_{i}$ is symmetric positive--definite. Combining the above three definitions
with Eq. \eqref{eqpdfXiXjConstr} and expressing $\mathbf{A}_{i}$, $\vec{b}_{i}$, $\mathbf{C}_{i}$, 
$\vec{d}_{i}$, $\mathbf{E}_{i}$, $f_{i}$ in terms of $\vec{\omega}_{i}$, $\mathbf{\Phi}_{i}$ and $\mathbf{V}_{i}$, 
we obtain again Eq. \eqref{eqAbCdEfVoP}. Then, we can follow the equivalences in backward direction 
(Eqs. \eqref{eqAbCdEfVoP}$\rightarrow$\eqref{eqmvnorm2ndForm}$\rightarrow$\eqref{eqmvnorm}) 
to prove that the pdf defined in Eq. \eqref{pdfXiXj} is equivalent to the Gaussian pdf 
defined in terms of $\vec{\omega}_{i}$, $\mathbf{\Phi}_{i}$ and $\mathbf{V}_{i}$, Eq. \eqref{eqmvnorm}. 
We note also that $\vec{\omega}_{i}$, $\mathbf{\Phi}_{i}$ and $\mathbf{V}_{i}$ defined above are constant 
with respect to $\vec{x}_{j}$, because they are defined in terms of $\mathbf{A}_{i}$, $\vec{b}_{i}$ and $\mathbf{E}_{i}$, 
which are constant with respect to $\vec{x}_{j}$ by definition. With that we proved condition $2$ 
of Def. \ref{defGLinv}. Since $\mathcal{M}$ satisfies condition $1$ of Def. \ref{defGLinv} 
by the first sentence in the Theorem, it follows that $\mathcal{M}$ belongs to the $\mathcal{G}_{LInv}$--family.
\end{proof}
\end{theorem}

\begin{remark}
In Eq. \eqref{pdfXiXj}, it suffices to consider symmetric negative--definite matrices $\mathbf{A}$ only. 
We remind that, by definition, a matrix $\mathbf{A}$ is negative--definite iff $\vec{x}^T\mathbf{A}\vec{x}<0$ 
for every $\vec{x}\neq\vec{0}$. Considering non--symmetric negative--definite matrices $\mathbf{A}$ does 
not extend the family of pdfs represented by Eq. \eqref{pdfXiXj}. In particular, for any square 
negative--definite matrix $\mathbf{Q}$, and (of appropriate size) vector $\vec{u}$, it 
holds that $\vec{u}^{T}\mathbf{Q}\vec{u}=\vec{u}^{T}\big[\frac{1}{2}(\mathbf{Q}+\mathbf{Q}^{T})\big]\vec{u}$ 
and the matrix $\big[\frac{1}{2}(\mathbf{Q}+\mathbf{Q}^{T})\big]$ is symmetric negative--definite. 
Hence if one took in Eq. \eqref{pdfXiXj} a non--symmetric $\mathbf{A}$, then the value of the pdf would be the 
same as if one had taken the symmetric negative--definite matrix $\big[\frac{1}{2}(\mathbf{A}+\mathbf{A}^{T})\big]$. 
\del{
Furthermore, considering non--negative--definite matrices $\mathbf{A}_{i}$ does not extend the family of pdfs 
represented by Eq. \ref{pdfXiXj} either. If one considers a non--negative--definite matrix $\mathbf{A}_{i}$, 
Eq. \eqref{pdfXiXj} could not define a probability density function, because the matrix $(-2)\mathbf{A}_{i}$ 
in Eq. \eqref{eqIntegralPdfXiXj} would not be positive--definite, and the integral over the domain of $\vec{x}_{i}$ 
will be infinite.
}
\KBcom{Not true. If we drop the assumption of non-zero support and take compact support, then we can have pos-def As. 
Of course we will move then out of the GLInv family and the model will not be normal. This is a very delicate
matter and hence I suggest we drop the whole remark. Since we present the model through omega, Phi, V then
the AbCdEf are presented more as a computational approach and so the reason for the remark is unclear. 
I suggest dropping it.
}\VMcom{I accept the deletion. I still think that the remark as it is written now is quite clear, so I suggest including it. If there is a simpler formulation of the possibility to have pos-def As by omitting non-zero support requirement, we can extend the remark later on.}
\end{remark}

Based on the above theorem and remark, we conclude that the $\mathcal{G}_{LInv}$-family is identical with the 
scope of the fast likelihood computation framework. This implies that to define any new model within the framework, 
it is sufficient to define the functions $\vec{\omega}$, $\mathbf{\Phi}$ and $\mathbf{V}$ for each  in the tree. 
This is the key idea in developing the  R-package \pkg{PCMBase} described in the next section.

\section[Software]{The \pkg{PCMBase} \proglang{R} package}\label{secSoftware}
The \pkg{PCMBase} package takes advantage of the fact that the quadratic polynomial representation of the 
likelihood function 
is valid for all models in the $\mathcal{G}_{LInv}$ family. Hence, once the analytical integration over the 
internal nodes has 
been implemented, the addition of a new $\mathcal{G}_{LInv}$ model to the framework boils down to defining the transition 
density in terms of the functions $\vec{\omega}$, $\mathbf{\Phi}$ and $\mathbf{V}$ (Def. \ref{defGLinv}).  
\pkg{PCMBase} implements this idea, based on the concept of inheritance between programming modules: 
Eqs. \eqref{eqAbCdEfVoP}, \eqref{pdfLmn_alltips}, \eqref{pdfLmn_tips}, \eqref{pdfLmn_non-tips} are implemented in a 
base module called ``GaussianPCM'', which is abstract with respect to $\vec{\omega}$, $\mathbf{\Phi}$ and $\mathbf{V}$ 
(Fig. \ref{figPCMBaseDiagram}). These functions are provided in inheriting  modules definable for each $\mathcal{G}_{LInv}$ 
model. This hierarchical design is presented in Fig. \ref{figPCMBaseDiagram}. 

\begin{figure}
\begin{center}
\includegraphics[width=1\textwidth]{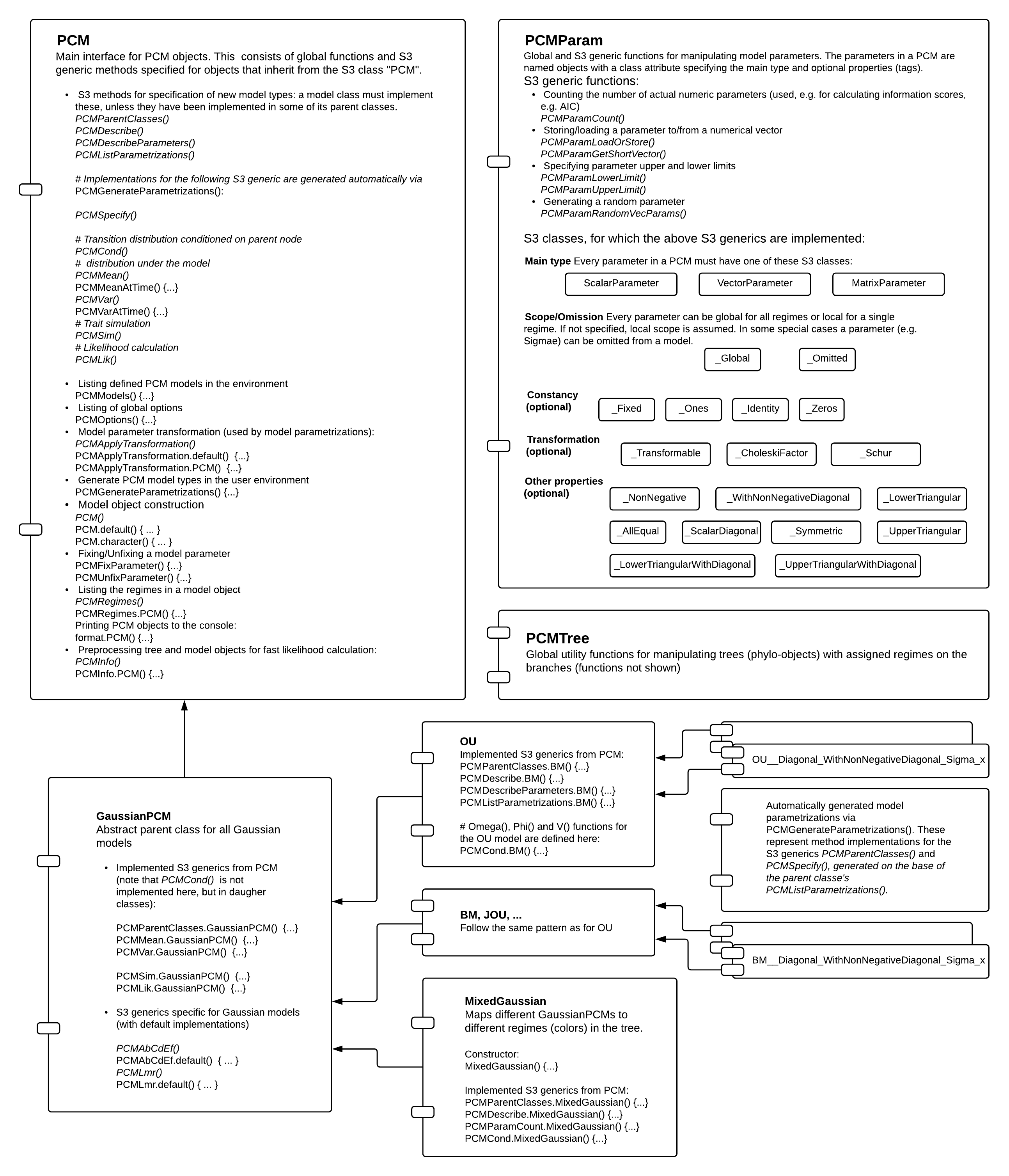}
\end{center}
\caption{\textbf{An overview of the PCMBase package.} Each box represents a module. The modules ``PCM'', 
``PCMParam'' and ``PCMTree'' define the end-user interface. In particular, the module ``PCM'' 
defines the interface for adding model extensions. Function names written in \textit{italic} 
style denote S3 generic declarations. These functions can be defined or overwritten by inheriting modules, 
to provide model-specific behavior. The module ``GaussianPCM'' implements the pruning-wise likelihood evaluation. 
The functions $\vec{\omega}$, $\mathbf{\Phi}$ and $\mathbf{V}$ for each model within the framework 
must be implemented in specifications of the S3 generic function ``PCMCond''. It is possible to define 
parametrizations restricting particular model parameters, e.g. forcing a matrix parameter to be a diagonal matrix.}
\label{figPCMBaseDiagram}
\end{figure}

\subsection[PCMBaseExtend]{Extending \pkg{PCMBase}}\label{subSecPCMBaseExtend}

Extending the \pkg{PCMBase} functionality can be achieved in two ways:
\begin{enumerate}
\item \textbf{Adding a new model}. It is possible to write a new module inheriting from the module ``GaussianPCM'' 
and implementing its own version of the functions $\vec{\omega}$, $\mathbf{\Phi}$ and $\mathbf{V}$;
\item \textbf{Adding a parameterization}. It is possible to restrict or apply a transformation to some of the parameters 
of an already defined model (Fig. \ref{figPCMBaseDiagram}). 
\end{enumerate}

\subsection[PCMBaseUse]{Using the package}\label{subSecPCMBaseUse}

Figure \ref{figPCMBaseCheatSheet} shows the runtime objects and use-cases currently implemented in the \pkg{PCMBase} 
package. 
\begin{figure}
\begin{center}
\includegraphics[width=.9\textwidth]{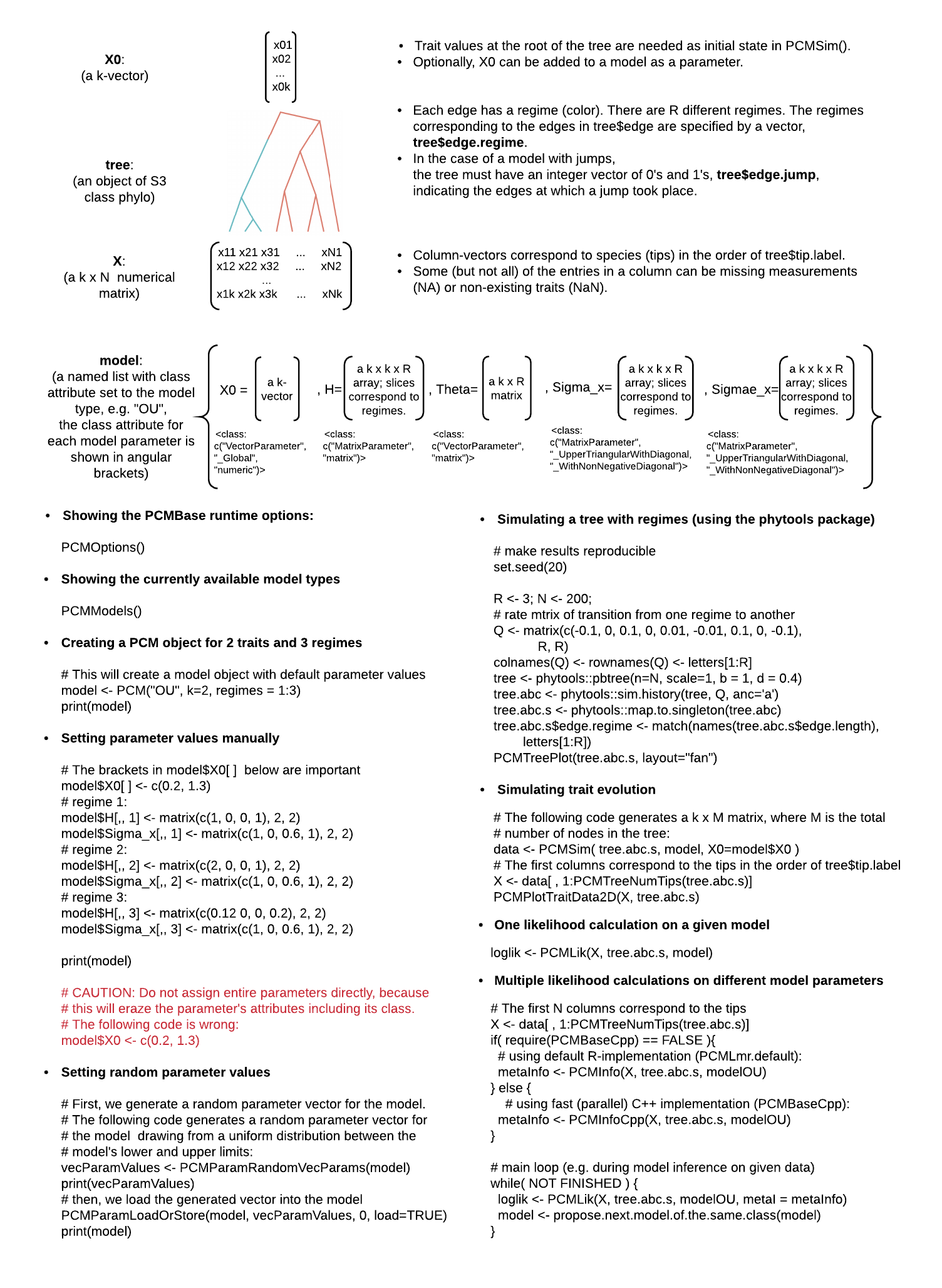}
\end{center}
\caption{\textbf{Using the \pkg{PCMBase} package} The main runtime objects are depicted on the top of the figure, 
followed by coding examples for the specific use-cases.}
\label{figPCMBaseCheatSheet}
\end{figure}
Once the modules for the models of interest have been implemented, the \pkg{PCMBase} package can be used to:
\begin{itemize}
\item Creating a model object. The end-user function for creating a model object is \code{PCM()}. A model object 
represents an  S3 object, 
that is, a named list with members corresponding to the model parameters, such as 
\code{H}, \code{Sigma\_x} and \code{Sigmae\_x}, and a class attribute equalling the model type, e.g. \code{BM} or \code{OU}. 
\item Simulating the evolution of a set of 
continuous traits along a tree, according to a model. The user level function for trait simulation is \code{PCMSim()}. 
Based on the S3 class of its model argument 
\code{PCMSim()} invokes an appropriate specification of the S3 generic function \code{PCMCond()}, which creates a 
random sampler 
from the trait distribution at the end of a branch, given the model, the branch length and the trait values 
at the beginning of the branch.  
\item Calculating the (log--)likelihood of a model, 
given a tree and trait values at its tips. The user level function for likelihood calculation is \code{PCMLik()}. 
This function is implemented in the ``GaussianPCM'' module and inherited by all of its daughter modules. 
The calculation proceeds in four steps: 
\begin{enumerate}
\item Initially, the model-specific functions $\vec{\omega}$, $\mathbf{\Phi}$ and $\mathbf{V}$ are calculated based
on the model parameters $\Theta$ and the branch lengths $t_i$ (note that this operation does not need the trait values to 
be present at any tip or internal node in the tree). 
\item Then, the coefficients $\mathbf{A}_{i}$, $\vec{b}_{i}$, $\mathbf{C}_{i}$, $\vec{d}_{i}$, $\mathbf{E}_{i}$ 
and $f_{i}$ are calculated for each internal and tip node in the tree based on the values $\vec{\omega}$, $\mathbf{\Phi}$ 
and $\mathbf{V}$ calculated in the previous step. This calculation is done in the function \code{PCMAbCdEf()} 
within the module 
``GaussianPCM'' which, again, is inherited by all model modules (see Fig. \ref{figPCMBaseDiagram}). 
\item Next, the coefficients $\mathbf{L}_{i}$, $\vec{m}_{i}$, $r_{i}$ 
are calculated based on the trait values at the tips, the values of $\mathbf{A}_{i}$, $\vec{b}_{i}$, $\mathbf{C}_{i}$, 
$\vec{d}_{i}$, $\mathbf{E}_{i}$ 
and $f_{i}$ calculated in the previous step, and the recursive procedure described in Section \ref{subsecmvFast}, 
Eqs. \eqref{pdfLmn_alltips}, \eqref{pdfLmn_tips} and \eqref{pdfLmn_non-tips}. 
\item Finally, the (log--)likelihood value is calculated using the formula
\be\label{pdfXX0X0}
\ell(\Theta)=pdf(\mathbf{X}|\vec{x}_{0},\mathbb{T},\mathbf{\Theta})=\exp\left(\vec{x}_{0}^{T}\mathbf{L}_{0}\vec{x}_{0}+
\vec{x}_{0}^{T}\vec{m}_{0}+r_{0}\right),
\ee
\noindent where $\Theta$ denotes the set of model parameters and $\vec{x}_{0}$ is 
treated either as a parameter (specified as a member 
\code{X0} in the model object) or as 
the optimum point of the above equation given by:
\be\label{optimumX0}
\vec{x}_{0}=-0.5 L_{0}^{-1} \vec{m}_{0}.
\ee
\end{enumerate}
\end{itemize}

The \pkg{PCMBase}-package is licensed under the General Public Licence (GPL) version 3.0. The package and the documentation are accessible from https://github.com/venelin/PCMBase.

\subsection[PCMBaseCpp]{Parallel likelihood calculation with the \pkg{PCMBaseCpp} add--in}\label{subsecPCMBaseCpp}
For faster likelihood calculation, it is possible to use multiple processor cores to perform the calculation of 
$\vec{\omega}$, $\mathbf{\Phi}$, $\mathbf{V}$, $\mathbf{A}_{i}$, $\vec{b}_{i}$, $\mathbf{C}_{i}$, $\vec{d}_{i}$,
$\mathbf{E}_{i}$ and $f_{i}$ in parallel. This is \del{legal}\TS{possible}, given the fact that these coefficients depend solely on the 
model parameters and on the branch lengths in the tree, see, e.g. Eqs. \eqref{eqmvOUmoments} and \eqref{eqmvOUVoP}. 
The calculation of the coefficients $\mathbf{L}_{i}$, $\vec{m}_{i}$, $r_{i}$ is not fully parallelizable but can be 
divided in parallelizable steps (generations) using a parallel post--order traversal algorithm  \citep{Mitov:2017eg}. 
We implemented  this idea in the accompanying package \pkg{PCMBaseCpp},  built on top of the \pkg{Armadillo} template 
library for linear algebra \citep{Sanderson:2016cs}, the \pkg{Rcpp} package for seamless \proglang{R} and \proglang{C++} 
integration \citep{Eddelbuettel:2013if} and the \pkg{SPLITT} library for parallel tree traversal \citep{Mitov:2017eg}. 

We compared the performance of the multivariate serial and parallel \pkg{PCMBase} implementation against other univariate and
multivariate implementations in a separate work \citep{Mitov:2017eg}. As shown in \citep{Mitov:2017eg}, 
on contemporary multi-core CPUs, the parallel \pkg{PCMBaseCpp} implementation can speed up the likelihood 
calculation up to an order of magnitude starting with 2 traits and trees of 100 to 10'000 tips. For univariate OU models, 
it can be beneficial to implement stand-alone classes bypassing the complex $k\times k$ matrix operations involved in 
the multivariate case. As shown in \citep{Mitov:2017eg}, this can result in up to $100$ 
fold faster likelihood calculation in the stand-alone class implementation. The use of \pkg{PCMBaseCpp} as a 
C++ back-end is recommended even if not using multi-core parallelization, because serial \proglang{C++} code 
execution is still nearly $100$ times faster than the equivalent implementation written in \proglang{R} 
(\proglang{R}--version at time of writing this article was $3.5$).

The \pkg{PCMBaseCpp}-package is licensed under the General Public Licence (GPL) version 3.0. The package is accessible from https://github.com/venelin/PCMBase.

\section[Standard extensions]{Standard extensions}\label{secSpecialIssues}

\subsection[Missing values]{Missing values}\label{sbsecNAs}
The trait measurement data are the observations at the tips. If a tip is described by a suite of traits it can 
easily happen that some of them are missing, either due to missing measurement or because the corresponding trait 
does not exist for the species. 
Removing such a tip from any further analysis would be wasting information, i.e. the observed data for the tip. 
We notice that missing measurements for existing traits correspond to the marginal distribution of the observed 
measurements. In contrast, non-existing traits correspond to reduced dimensionality of the trait vector for the tip in 
question. Our computational framework keeps track of both of these cases by carefully accounting for the dimensionality 
of the trait vectors at the tips and the internal nodes and appropriately marginalizing during the integration part, as 
described below (see also Thm. \ref{thmmvOU} for examples). The input data is passed as a matrix 
(rows---trait measurements, columns---different species) the missing measurements have to be indicated as \code{NA}s, 
whereas the non-existing traits have to be indicated as \code{NaN}s (fig. \ref{figTreeNotation}).


We now turn to describing the technicalities of the mechanism taking care of the missing data.  
We use a vector of positive integers, $\vec{k}_{j}$, to denote the ordered set of active coordinates for a node $j$.
If $j$ is a tip, then $\vec{k}_{j}$ gives the indices of all non--missing entries in the trait vector for $j$; 
for an internal (unmeasured) node this gives the possibility to make some trait inactive.
The cardinality of a vector is denoted with $\vert \vec{k} \vert$. 
For a vector, the notation $\vec{\theta}[\vec{k}]$
means the vector of elements of $\vec{\theta}$ on the coordinates contained in $\vec{k}$, while
for a matrix $\mathbf{H}[\vec{k}_{1},\vec{k}_{2}]$ means the matrix $\mathbf{H}$ with only the rows 
on the coordinates contained in $\vec{k}_{1}$ and columns contained in $\vec{k}_{2}$. For example
take $\vec{\theta}=(10,11,12,13)$ and $\vec{k}=(1,3)$, then $\vec{\theta}[\vec{k}]=(10,12)$,
while if $\vec{k}_{1}=(1,3)$, $\vec{k}_{2}=(2,4)$ and 

$$
\mathbf{H}=\left[
\begin{array}{cccc}
10 & 11 & 12 & 13 \\
14 & 15 & 16 & 17 \\
18 & 19 & 20 & 21 \\
22 & 23 & 24 & 25
\end{array}
\right],
$$
then 

$$
\mathbf{H}[\vec{k}_{1},\vec{k}_{2}]=\left[
\begin{array}{cccc}
11 & 13 \\
19 & 21 \\
\end{array}
\right].
$$
If a vector or matrix does not have any indication on which entries it is retained, then it means
that we use the whole vector or matrix.
All of the above notation is graphically represented in Fig. \ref{figTreeNotation}.

In our framework, we have the representation that
$\vec{x}_{i} \in \mathbb{R}^{k_{i}}$ conditional on $\vec{x}_{j} \in \mathbb{R}^{k_{j}}$
is $\mathcal{N}(\vec{\omega}_{i} + \mathbf{\Phi}_{i}\vec{x}_{j},\mathbf{V}_{i})$ distributed. By default, \pkg{PCMBase} constructs the coordinate vectors $\vec{k}_{i}$ and $\vec{k}_{j}$ in the following way: for a tip-node, $i$, $\vec{k}_{i}$ contains all observed (neither \code{NA} nor \code{NaN}) 
coordinates; 
for an internal node, $i$ or $j$, the corresponding coordinate vector ($\vec{k}_{i}$ or $\vec{k}_{j}$) contains the coordinates denoting traits that exist (are not \code{NaN}) for at least one of the 
tips descending from that node (Fig. \ref{figTreeNotation}). Biologically, this treatment reflects a scenario where all of the traits with at least one non-\code{NaN} entry for at least one species (i.e. tip) in the tree must have existed for the root but some of the traits have subsequently disappeared on some lineages of the tree. In particular, if a trait exists for a given tip in the tree, it is assumed that it has existed for all of its ancestors up to the root of the tree. Conversely, if the trait does not exist for a tip, then it has not existed for any of its ancestors up to the first ancestor shared with a tip for which the trait does exist. Different biological scenarios are possible, e.g. assuming that some of the traits did not exist at the root-node but have appeared later for some on the lineages. These can be implemented by accordingly specifying the coordinate vectors.

During likelihood calculation for given trait data, a tree and a trait evolutionary model, the elements $\vec{\omega}_{i}$, $\mathbf{\Phi}_{i}$ and $\mathbf{V}_{i}$ of appropriate dimension are calculated for each non-root node $i$ in the tree. This is done in two steps:
\begin{enumerate}
\item The general rule of the model is used to calculate the elements $\tilde{\vec{\omega}}_{i}$, $\tilde{\mathbf{\Phi}}_{i}$, $\tilde{\mathbf{V}}_{i}$ of full dimensionality ($k$), i.e. assuming that all traits exist;
\item Denoting by $j$ the parent node of $i$, the elements $\vec{\omega}_{i}$, $\mathbf{\Phi}_{i}$ and $\mathbf{V}_{i}$ specific for the data in question are obtained as:
\be\label{eqTreatNAs}
\begin{array}{rcl}
\vec{\omega}_{i} &=& \tilde{\vec{\omega}}_{i}[\vec{k}_{i}], \\
 \mathbf{\Phi}_{i}&=&  \tilde{\mathbf{\Phi}}_{i}[\vec{k}_{i},\vec{k}_{j}], \\
\mathbf{V}_{i}&=& \tilde{\mathbf{V}}_{i}[\vec{k}_{i},\vec{k}_{i}],
\end{array}
\ee
\noindent where $\vec{k}_{i}$ and $\vec{k}_{j}$ denote the corresponding coordinate vectors at nodes $i$ and $j$.
\end{enumerate}

%
%

\subsection[Measurement error]{Measurement error}\label{sbsecMerror}
Commonly in PCMs the observed values at the tips are averages 
from a number of individuals of each species. Using just these
average values does not take into account the intra--species variability.
Ignoring this can have profound effects on any further estimation
\citep[see][]{Hansen:2012et}. Following the PCM tradition, we call
this intra--species variability a\del{s} measurement error,
but one should remember that it can be due to true biological variability.
Including this variability in our framework is straightforward. One recognizes,
which component of the quadratic polynomial representation
corresponds to the variance of the tip and augments it by the measurement error variance
matrix, see the formulae in Section \ref{secpcmOU}. From the user interface point
of view this is a bit more complicated. The measurement error variance
matrix is specific to each tip. Therefore in this situation the user has to define
for each tip a different regime, with a regime specific variance matrix
(called \code{Sigmae\_x} in the implemented by us classes). Of course other model
parameters can also be regime specific, e.g. the deterministic optima
(\code{Theta} in the implemented by us classes). 

\subsection[Non--ultrametric trees and multifurcations]{Non--ultrametric trees and multifurcations}\label{sbsecNonUltra}
If one has only measurements from contemporary species, then the phylogeny
describing them is naturally an ultrametric one. However,  if for some 
reason the phylogeny is not ultrametric, e.g. it contains extinct species,
then the quadratic polynomial framework can be directly employed. 
Because each branch is treated separately, it does not matter whether 
the tree is or is not ultrametric. Therefore, there is no need
to search for transformations as in the $3$--point structure based methods.
This we believe should make the \pkg{PCMBase} package very straightforward to use. 
Furthermore, from the proof of Thm. \ref{thmMvFast} it is obvious that 
the tree does not need to be binary. Therefore, this adds even more flexibility
to the user, they may use trees with polytomies.

\subsection[Punctuated equilibrium]{Punctuated equilibrium}\label{sbsecJumps}
It is an ongoing debate in evolutionary biology whether 
the dominant mode of evolution is a gradual one or whether, during
brief periods of time, species undergo rapid change. Any gradual 
model of evolution can be extended to have a punctuated component
by including jumps. 
Jump mechanisms, like jumps at the start of specific lineages or common jumps for daughter lineages,
have to be developed on a per model basis, see Section 
\ref{sbsecOUjumps} for an example.
One current restriction is that \pkg{PCMBase}
assumes that lineages do not interact after speciation. 
It is not possible to implement a model class such that
if one daughter lineage jumps the other does not (this is communication
between lineages after speciation). Therefore, to have such
a situation the user needs to by themselves code on which 
lineages a jump can take place and on which it cannot.
This can be easily achieved using the jumps mechanism of \pkg{PCMBase}. 
The \code{phylo} phylogenetic tree object can be enhanced by a
\code{edge.jump} binary vector. The length of this vector equals the 
number of edges in the tree. A $0$ entry indicates that no jump
took place on the corresponding branch, while a $1$ entry that it did. 

\section[Ornstein--Uhlenbeck type models]{Ornstein--Uhlenbeck type models}\label{secpcmOU}
\begin{figure}
\begin{center}
\includegraphics[width=1\textwidth]{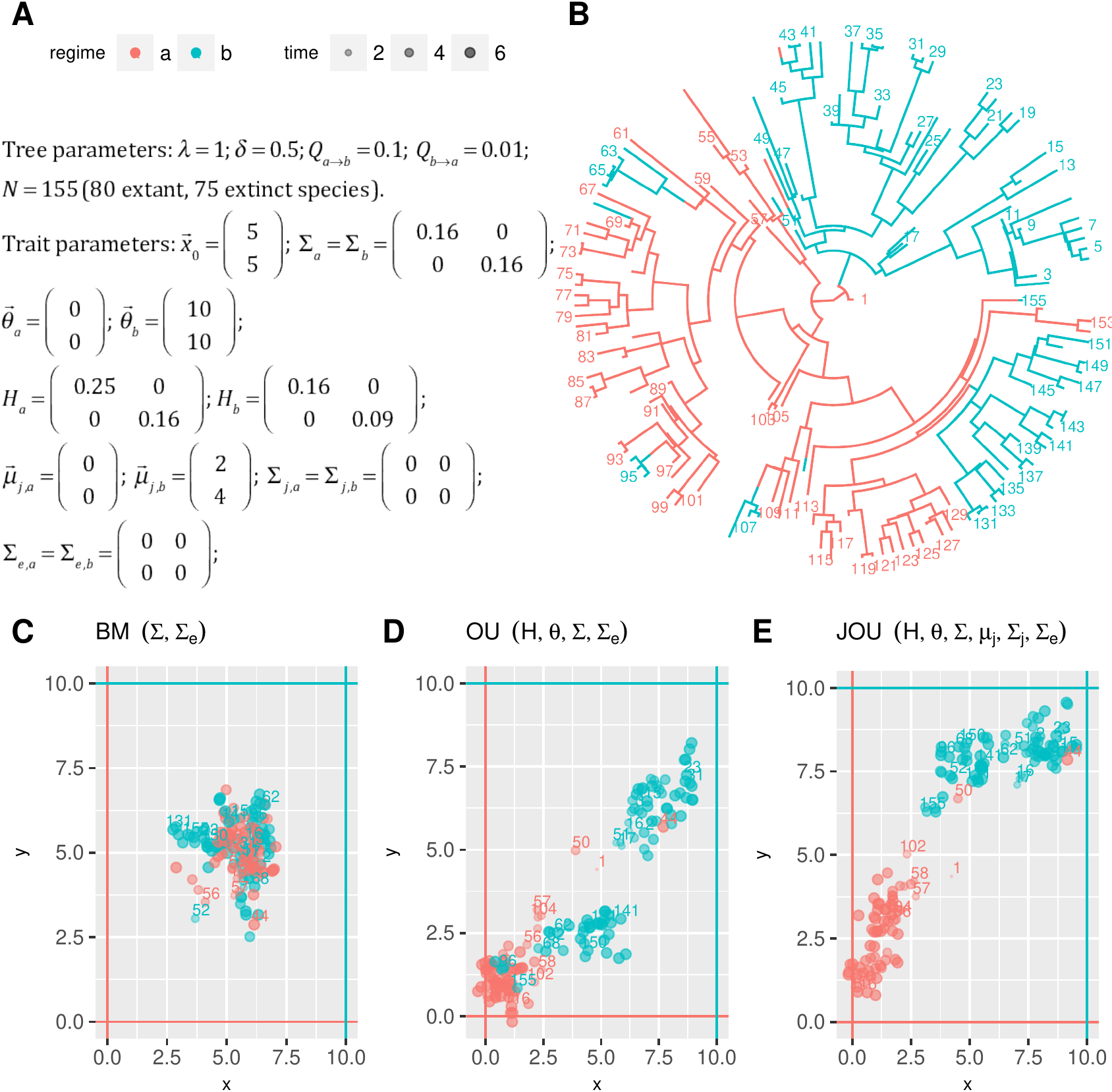}
\end{center}
\caption{Simulations of trait evolution under four PCMs. 
A: parameters of the tree simulation ($\lambda$: speciation--rate; $\delta$: extinction--rate; 
$Q_{a\rightarrow b}$: migration rate from habitat ``a'' to habitat ``b''; $Q_{b\rightarrow a}$: 
vice--versa of $Q_{a\rightarrow b}$. The other parameters are described in the text. 
B: A birth--death phylogenetic tree generated using the function \code{pbtree()} 
and \code{sim.history()} from the package \pkg{phytools} \citep{Revell:2011ku}. 
C---E: scatter plots of the traits observed at the tips of the tree after random simulation using the 
function \code{PCMSim()} of the \pkg{PCMBase} package.}
\label{figPCMsim}
\end{figure}

\subsection[The phylogenetic Ornstein--Uhlenbeck process]{The phylogenetic Ornstein--Uhlenbeck process}
Currently the Ornstein--Uhlenbeck process is the workhorse of the
phylogenetic comparative methods framework. 
Since its introduction by \citet{Hansen:1997ek} it has been considered in detail
with multiple software implementations
\citep[e.g.][to name a few]{Bartoszek:2012fw,Beaulieu:2012ex,Butler:2004ce,Clavel:2015hc,FitzJohn:2010ch,Goolsby:2016fp,Hansen:2008gt,Ho:2014ge}

In the most general form, the multivariate Ornstein--Uhlenbeck process describes the
evolution of a $k$--dimensional suite of traits $\vec{x}\in \mathbb{R}^{k}$ over
a period of time by the following stochastic differential equation 

\be\label{eqmvOU}
\ud \vec{x}(t) = -\mathbf{H}\left(\vec{x}(t)-\vec{\theta}(t) \right)\ud t + \mathbf{\Sigma}_{x}\ud \vec{W}(t),
\ee
$\mathbf{H}\in \mathbb{R}^{k\times k}$, $\vec{\theta}(t)\in \mathbb{R}^{k}$
and $\mathbf{\Sigma}_{x}\in \mathbb{R}^{k\times k}$.
Notice that when $\mathbf{H}=\mathbf{0}$, we obtain a Brownian motion model.

There is no current software package, in the case of phylogenetic OU models,
that allows for an arbitrary form of the matrix $\mathbf{H}$.
Except for the Brownian motion case, nearly all
assume that $\mathbf{H}$ has to be symmetric--positive--definite (note that
this encompasses the single trait case).
\pkg{mvMORPH} \citep{Clavel:2015hc}, \pkg{SLOUCH} \citep{Hansen:2008gt}
and \pkg{mvSLOUCH} \citep{Bartoszek:2012fw} seem to be the only
exceptions. \pkg{mvMORPH} and \pkg{mvSLOUCH} allow for a general invertible
$\mathbf{H}$ 
(with options to restrict it to diagonal, triangular, symmetric positive--definite,
positive eigenvalues, real eigenvalues or generally invertible).
Furthermore, \pkg{mvSLOUCH} allows for a special singular structure of $\mathbf{H}$.
The matrix has to have in the upper--left--hand corner an invertible matrix (\pkg{SLOUCH}, the univariate
predecessor of \pkg{mvSLOUCH} has a scalar here),
arbitrary values to the right and $\mathbf{0}$ below.
This type of model is called an Ornstein--Uhlenbeck--Brownian motion (OUBM) model.
In contrast, when $\mathbf{H}$ is non--singular the model is called an
Ornstein--Uhlenbeck--Ornstein--Uhlenbeck (OUOU) one, some variables
are labelled as predictors while the rest as responses.

It is of course not satisfactory to have restrictions
on the form of $\mathbf{H}$. Different setups
have different biological interpretations with regards
to modelling causation \citep[see][]{Bartoszek:2012fw,Reitan:2012hi}.
In particular singular matrices will be interesting as they will correspond
to certain linear combinations of traits under selection pressures while other linear combinations are
free of this. The OUBM model is a special case where a pre--defined group
of traits is assumed to evolve marginally as a Brownian motion.
Of course a more general setup is desirable and actually, as we show in this
work, possible.

Here the the only assumption we make on $\mathbf{H}$ is that it possesses an eigendecomposition,
$\mathbf{H}=\mathbf{P}\mathbf{\Lambda}\mathbf{P}^{-1}$
($\mathbf{\Lambda}$ is a diagonal matrix, and the $i$--th element of the diagonal is denoted as $\lambda_{i}$).
In particular $\mathbf{\Lambda}$ can be singular, i.e. some eigenvalues are $0$ and furthermore the
eigenvalues/eigenvectors are allowed to be complex.

In this work we assume that $\mathbf{\Sigma}_{x}$ is upper triangular (alternatively lower).
Despite how it looks at first sight,
this is not any sort of restriction, as in the likelihood
we have only $\mathbf{\Sigma}:=\mathbf{\Sigma}_{x}\mathbf{\Sigma}_{x}^{T}$.
We furthermore assume that  $\mathbf{\Sigma}$ is non--singular,
otherwise the whole model would be singular from a statistics point of view.

Most OU model implementations assume that the deterministic optimum $\vec{\theta}_{t}$
is constant along each branch. Different branches may have different levels
of it but regime switches along a branch are not allowed
\citep[however,][are exceptions as they attempt to infer points of switching]{Bastide:2018eq,Ingram:2013cm,Khabbazian:2016fr}.
The \pkg{PCMBase} package does not make any inference but allows for regime switches inside a branch,
in the sense that the user (or software using \pkg{PCMBase}'s functionality) has to split the branch
into branches having constant regimes. 


If we assume that the process starts at a value $\vec{x}(0)=\vec{x}_{0}$,
then after evolution over time $t$ (assuming all parameters are constant on this interval)
it will be normally distributed
with mean vector and variance--covariance matrix \citep[Eqs. (A.1, B.2)][]{Bartoszek:2012fw}

\be
\begin{array}{rcl}\label{eqmvOUmoments}
\Expectation{\vec{x}} (t) & = & e^{-\mathbf{H}t} \vec{x}_{0} + \left(\mathbf{I}- e^{-\mathbf{H}t}\right) \vec{\theta} \in \mathbb{R}^{k}\\
\variance{\vec{x}} (t) & = & \int_{0}^{t} e^{-\mathbf{H}v} \mathbf{\Sigma}e^{-\mathbf{H}^{T}v} \ud v
\\ & = &
\mathbf{P}\left(
\left[ \frac{1}{\lambda_{i}+\lambda_{j}}\left(1-e^{-(\lambda_{i}+\lambda_{j})t}\right)
\right]_{1\le i,j \le k} \odot \mathbf{P}^{-1}\mathbf{\Sigma}\mathbf{P}^{-T}
\right)\mathbf{P}^{T} \equiv \mathbf{V}(t) \in \mathbb{R}^{k \times k},
\end{array}
\ee
where $\mathbf{I}$ is the identity matrix of appropriate size.
Notice that in the above, $\mathbf{H}$ only enters the moments, through its exponential.
Therefore the moments can be calculated (and hence the distribution is well defined)
for all $\mathbf{H}$, including defective ones.
However, if $\mathbf{H}$ has (as we assumed) an eigendecomposition, then the exponential and in turn
variance formula can be calculated effectively. If $\lambda_{i} = \lambda_{j}=0$, then
the term in the variance has to be treated in the limiting sense
$\lambda^{-1}(1-e^{-\lambda t}) \to t$ with $\lambda \to 0$. Therefore, the variance matrix
is always well defined and never singular for $t>0$.

We assumed that $\mathbf{H}$ has to have an eigendecomposition while the process
is well defined for any $\mathbf{H}$, including defective ones. Calculation
of the matrix exponential for a defective matrix can be done using Jordan block decomposition.
However, we do not provide such functionality, as Jordan block decomposition is numerically unstable and
in fact, we are not aware of any \proglang{R} implementation of it. 
Hence, defective matrices will result in errors. 
However, it is important to remember that defectiveness is the exception and not the rule for matrices. If 
checked for \citep[by e.g. checking if the eigenvector matrix from \code{eigen()}'s output is non--singular, Corollary $7$.$1$.$8$., p. $353$][]{Golub:2013mc}
and handled before calling using our package, it should not cause major issues.
\TScom{why don't we state what $V,\Psi, omega$ is?}\KBcom{This is done in the next section}


We now turn to showing how to construct
the composite parameters found in the proof of Thm. \ref{thmpcmNorm} from the OU process representation
of Eq. \eqref{eqmvOU}. 

To simplify notation we denote the defined in Eq. \eqref{eqmvOUmoments} covariance matrix as
$\tilde{\mathbf{V}}_{i} \equiv \mathbf{V}(t_{i})+
\delta_{i\in\{\mathbb{T}_{0}'s~ tips\}}\mathbf{\Sigma}_{e}^{i}$, 
where the Kronecker $\delta$--symbol is defined as $\delta_{i\in\{\mathbb{T}_{0}'s~ tips\}}=1$ 
if $i$ is a tip of the tree and $\delta_{i\in\{\mathbb{T}_{0}'s~ tips\}}=0$. The 
matrix $\mathbf{\Sigma}_{e}^{i}$ is the measurement error or intra--species variability variance matrix
for tip species $i$.

\begin{theorem}\label{thmmvOU}
Let $\vec{k}_{i}$ be the vector of coordinates on which $\vec{x}_{i}$ is observed, $\vec{k}_{j}$
be the vector of coordinates for $\vec{x}_{j}$ and $\vec{k}$ the full vector of
coordinates.
Using the parameterization found in the proof of Thm. \ref{thmpcmNorm}
a multivariate Ornstein--Uhlenbeck process of evolution 
can be represented as

\be\label{eqmvOUVoP}
\begin{array}{rcl}
\mathbf{V}_{i} & = & \tilde{\mathbf{V}}_{i}[\vec{k}_{i},\vec{k}_{i}] \in \mathbb{R}^{\vert \vec{k}_{i} \vert \times \vert \vec{k}_{i} \vert },\\
\vec{\omega}_{i} & = & \left(\mathbf{I}[\vec{k}_{i},\vec{k}]-e^{-\mathbf{H} t_{i}}[\vec{k}_{i},\vec{k}]\right)\vec{\theta}_{i}[\vec{k}] \in \mathbb{R}^{\vert \vec{k}_{i} \vert},\\
\mathbf{\Phi}_{i} & = & e^{-\mathbf{H} t_{i}}[\vec{k}_{i},\vec{k}_{j}] \in \mathbb{R}^{\vert \vec{k}_{i} \vert \times \vert \vec{k}_{j} \vert }.
\end{array}
\ee
\end{theorem}

\begin{proof}
In the multivariate OU case, Eq. \eqref{pdfXiXj} will be

\bd
pdf(\vec{x}_{i}\vert \vec{x}_{j},t_{i})=
\mathcal{N}\left(
e^{-\mathbf{H} t_{i}}[\vec{k}_{i},\vec{k}_{j}]\vec{x}_{j} + 
\left(\mathbf{I}[\vec{k}_{i},\vec{k}]-e^{-\mathbf{H} t_{i}}[\vec{k}_{i},\vec{k}]\right)\vec{\theta}_{i}[\vec{k}],\mathbf{V}_{i}[\vec{k}_{i},\vec{k}_{i}]\right).
\ed

\end{proof}
These formulae do not depend on whether the eigenvalues of $\mathbf{H}$
are positive, negative or $0$. They will still be correct.
The exponentiation of $\mathbf{H}$ will also not depend on this. Only with
$\mathbf{V}_{i}$ will we need to take an appropriate limit as an eigenvalue is $0$,
see comments after Eq. \eqref{eqmvOUmoments}.

\begin{corollary}
For a multivariate Brownian motion process of evolution, we have
$\mathbf{H}=\mathbf{0}$ and 
$\tilde{\mathbf{V}}_{i}=t_{i}\mathbf{\Sigma} + \delta_{i\in\{\mathbb{T}_{0}'s~ tips\}}\mathbf{\Sigma}_{e}^{i}$.
Hence, using the parametrization found in the proof of Thm. \ref{thmpcmNorm}
one can represent it as

\be
\begin{array}{rcl}\label{eqmvBMVoP}
\mathbf{V}_{i}& = & \tilde{\mathbf{V}}_{i}[\vec{k}_{i},\vec{k}_{i}] \in \mathbb{R}^{\vert \vec{k}_{i} \vert \times \vert \vec{k}_{i} \vert },\\
\vec{\omega}_{i} & = & \vec{0}[\vec{k}_{i}] \in \mathbb{R}^{\vert \vec{k}_{i} \vert},\\
\mathbf{\Phi}_{i} & = & \mathbf{I}[\vec{k}_{i},\vec{k}_{j}] \in \mathbb{R}^{\vert \vec{k}_{i} \vert \times \vert \vec{k}_{j} \vert }.
\end{array}
\ee
\end{corollary}
In Fig. \ref{figPCMsim}, panel C, one can see an example collection of tip observations
resulting from simulating of a bivariate trait following a BM process on top of a phylogeny and in panel D following an OU process.

\subsection[Multivariate Ornstein--Uhlenbeck with jumps]{Multivariate Ornstein--Uhlenbeck with jumps}\label{sbsecOUjumps}
It is an ongoing debate in evolutionary biology at what time does evolutionary change take place. 
Change may take place either at times of speciation 
\citep[punctuated equilibrium][]{Eldredge:1972pe,Gould:1993iy}
or gradually accumulate \citep[phyletic gradualism, see references in][]{Eldredge:1972pe}. 
There seems to be evidence for both types of evolution.
For example, \citet{Bokma:2002hk} discusses that punctuated equilibrium is supported by fossil records \citep[see][]{Eldredge:1972pe}
but on the other hand \citet{Stebbins:1981ck} also indicate  
experiments supporting phyletic gradualism.

At an internal node in the tree something happens that drives species apart
and then ``The further removed in time a species from the original
speciation event that originated it, the more its genotype will have become stabilized and the more
it is likely to resist change.'' \citep{Mayr:1982et}. Between branching events (and jumps) we can 
have stasis---``fluctuations of little or no accumulated consequence'' taking place
\citep{Gould:1993iy}.
Therefore, one would want processes that incorporate both types of evolution and allow for testing if either
of them dominates. Ornstein--Uhlenbeck with jumps models are a framework where this is possible. 
Shortly, along a branch the traits follows an OU process. But then, just after speciation, a jump in the 
traits' values can take place. Whether such a jump takes place on a given, some or all daughter
lineages is up to the specific implementation of the framework. From the perspective of the 
\pkg{PCMBase} package the location of the jumps has to be provided. 
It is in fact also possible in our implementation, to place jumps at arbitrary points inside a branch
(this may be necessary if speciation events are missing due to unsampled extinct or extant species). 
Inferring of where jumps took place are at a different level of PCM modelling, then what
\pkg{PCMBase} handles.

Ornstein--Uhlenbeck processes with jumps capture the key idea behind the theory of punctuated
equilibrium. 
If the speed of convergence of the process
is large enough, then the stationary distribution is approached rapidly and the stationary oscillations around
the (constant) mean can be interpreted as stasis between jumps.

\begin{corollary}\label{corOUjumps}
For a multivariate OU 
defined with jumps, jump distribution $\mathcal{N}(\vec{\mu}_{J},\mathbf{\Sigma}_{J})$ and denoting
by the indicator $\xi_{i}$ (we assume that the jumps are known) if a jump took place at the start of the branch leading to node $i$, 
we have 

\be\label{eqmvJOUmoments}
\begin{array}{rcl}
\tilde{\mathbf{V}}_{i} & = & \int\limits_{0}^{t_{i}}e^{-\mathbf{H}v}\mathbf{\Sigma}e^{-\mathbf{H}^{T}v} \ud v + \xi_{i} e^{-\mathbf{H}t_{i}}\mathbf{\Sigma}_{J}e^{-\mathbf{H}^{T}t_{i}} + \delta_{i\in\{\mathbb{T}_{0}'s~ tips\}}\mathbf{\Sigma}_{e}^{i}.
\end{array}
\ee
Using the parametrization found in the proof of Thm. \ref{thmpcmNorm} one can represent it as

\be
\label{eqmvJOUVoP}
\begin{array}{rcl}
\mathbf{V}_{i} & = & \tilde{\mathbf{V}}_{i}[\vec{k}_{i},\vec{k}_{i}] \in \mathbb{R}^{\vert \vec{k}_{i} \vert \times \vert \vec{k}_{i} \vert },\\
\vec{\omega}_{i} & = & \xi_{i}e^{-\mathbf{H} t_{i}}[\vec{k}_{i},\vec{k}]\vec{\mu}_{J}[\vec{k}]+ \left(\mathbf{I}[\vec{k}_{i},\vec{k}]-e^{-\mathbf{H} t_{i}}[\vec{k}_{i},\vec{k}]\right)\vec{\theta}_{i}[\vec{k}] \in \mathbb{R}^{\vert \vec{k}_{i} \vert},\\
\mathbf{\Phi}_{i} & = & e^{-\mathbf{H} t_{i}}[\vec{k}_{i},\vec{k}_{j}] \in \mathbb{R}^{\vert \vec{k}_{i} \vert \times \vert \vec{k}_{j} \vert }.
\end{array}
\ee

\end{corollary}
The multivariate Brownian motion with jumps model follows as an immediate corollary $(\mathbf{H} \to \mathbf{0})$ .

\begin{corollary}
For a multivariate Brownian motion with jumps (jumps defined the same as in Corollary \ref{corOUjumps})
the variance at a node $i$ is $\tilde{\mathbf{V}}_{i}=t_{i}\mathbf{\Sigma}[\vec{k}_{i},\vec{k}_{i}]+\xi_{i} \mathbf{\Sigma}_{J}[\vec{k}_{i},\vec{k}_{i}] 
+ \delta_{i\in\{\mathbb{T}_{0}'s~ tips\}}\mathbf{\Sigma}_{e}^{i}$. 
Using the parametrization found in the proof of Thm. \ref{thmpcmNorm} one can represent it as

\be
\begin{array}{rcl}
\mathbf{V}_{i} & = & \tilde{\mathbf{V}}_{i}[\vec{k}_{i},\vec{k}_{i}] \in \mathbb{R}^{\vert \vec{k}_{i} \vert \times \vert \vec{k}_{i} \vert },\\
\vec{\omega}_{i} & = & \xi_{i}\vec{\mu}_{J}[\vec{k}_{i}] \in \mathbb{R}^{\vert \vec{k}_{i} \vert},\\
\mathbf{\Phi}_{i} & = & \mathbf{I}[\vec{k}_{i},\vec{k}_{j}] \in \mathbb{R}^{\vert \vec{k}_{i} \vert \times \vert \vec{k}_{j} \vert }.
\end{array}
\ee
\end{corollary}
In Fig. \ref{figPCMsim}, panel E, one can see an example collection of tip observations
resulting from simulating of a bivariate trait following an OU process 
with jumps on top of a phylogeny.

\subsection[Beyond the Ornstein--Uhlenbeck process]{Beyond the Ornstein--Uhlenbeck process}\label{sbsecBeyondOU}
There are a number of popular PCM models that do not fall into the above described OU framework despite
appearing very similar. In particular we mean the BM with trend, drift, early burst/Accelerating--decelerating (EB/ACDC)
or white noise \citep[implemented in the \pkg{geiger} \proglang{R} package][]{Harmon:2008iy}. 
With the exception of white noise, they all can be represented by the SDE 
\citep[cf. Eq. (1) of][]{Manceau:2016kz}
\VMcom{Is the matrix H below the same as the matrix H for the OU process? I saw that Manceau et al. uses the symbol A for the same parameter. Using the symbol H is confusing here. Please, use another one, or clarify that this parameter has the same meaning as the H in the OU process.}
\be
\left\{
\begin{array}{rcl}
\ud \vec{x}(t) & = & \left(\vec{h}(t)-\mathbf{H}\vec{x}(t) \right)\ud t + \mathbf{\Gamma}(t) \ud \vec{W}(t), \\
\vec{x}(0)  & = & \vec{x}_{0}.
\end{array}
\right.
\ee
Notice that setting $\vec{h}(t)=\vec{\theta}$ and $\mathbf{\Gamma}(t) = \mathbf{\Sigma}_{x}$ to constants we recover
 the ``usual'' OU process, considered in Thm. \ref{thmmvOU}.
\citet{Manceau:2016kz} provide the expectation and variance under the model, by slightly \VMcom{``slightly'' ? please, explain. Is it just notation change or something in the maths?}
modifying their Eqs. (4a) and (4b),

\be
\begin{array}{rcl}
\Expectation{\vec{x}_{i} \vert \vec{x}_{j}}  & = &
e^{-t_{i} \mathbf{H}_{i}}\vec{x}_{j} + \int\limits_{t_{i}^{s}}^{t_{i}^{e}}e^{(s-t_{i}^{e})\mathbf{H}_{i}}\vec{h}_{i}(s) \ud s, \\
\var{\vec{x}_{i} \vert \vec{x}_{j}}  & = &
\int\limits_{t_{i}^{s}}^{t_{i}^{e}}e^{(s-t_{i}^{e})\mathbf{H}_{i}}\mathbf{\Gamma}_{i}(s) \mathbf{\Gamma}_{i}^{T}(s) e^{(s-t_{i}^{e})\mathbf{H}_{i}^{T}}\ud s,
\end{array}
\ee
where $t_{i}^{s}$ is the time at the start of the branch and $t_{i}^{e}$ at the end (of course 
$t_{i}=t_{i}^{e} - t_{i}^{s}$).
This corresponds in our framework to 

\be\label{eqoPVManceau}
\begin{array}{rcl}
\vec{\omega}_{i} & = & \int\limits_{t_{i}^{s}}^{t_{i}^{e}}e^{(s-t_{i}^{e})\mathbf{H}_{i}}\vec{h}_{i}(s) \ud s, \\
\mathbf{\Phi}_{i} & = & e^{-t_{i} \mathbf{H}_{i}}, \\
\mathbf{V}_{i} & = & \int\limits_{t_{i}^{s}}^{t_{i}^{e}}e^{(s-t_{i}^{e})\mathbf{H}_{i}}\mathbf{\Gamma}_{i}(s) \mathbf{\Gamma}_{i}^{T}(s) e^{(s-t_{i}^{e})\mathbf{H}_{i}^{T}}\ud s.
\end{array}
\ee
\VMcom{From this paragraph, I cannot see how to implement the ACDC and the BM with trend models. 
What is the meaning of these models in biological terms? I remember that ACDC stays for acceleration/deceleration. 
What do BM with drift and BM with trend mean? Maybe split this paragraph into separate paragraphs for each of model?} 
\KBcom{In the first case the drift coefficient changes linearly with time, in the second the diffusion.
But what does this mean biologically? Depends on the system under study I guess, beyond that 
I cannot say much more. 
Below you have the equations for omega, Phi, V given for ACDC, BM with drift. 
For BM with drift it is written in the form that is compatible with all possibilities, as I cannot make out from geiger's
description what exactly they had in mind.
I do not think that we have enough to write for more than this fragment.
} \VMcom{Why citing geiger's user manual? This is just an implementation and not an original description of the models. For example, for ACDC, I've found that the original publication is the paper from Blomberg from 2003. Also for ACDC, there is no matrix parameter R, but a scalar parameter r in Manceau et. al. If we are doing a generalization from a scalar to a matrix, then this should be described clearly.}
Hence, \emph{in the subcase of non--interacting lineages}, our framework 
covers \citet{Manceau:2016kz}'s. As the initially mentioned models are subcases \citep[cf. Tab. 1 of][]{Manceau:2016kz}
they are available in our framework. To obtain the values of the $\vec{\omega}_{i}$, $\mathbf{\Phi}_{i}$
and $\mathbf{V}_{i}$ parameters in Eq. \eqref{eqoPVManceau} one has to either analytically calculate
the integrals for specific $\vec{h}_{i}(\cdot)$ and $\mathbf{\Gamma}(\cdot)$ functions
or consider a general numerical integration scheme.

Apart from the previously considered OU model, for some other typical PCM models the integrals can be evaluated analytically.
\begin{enumerate}
\item 
ACDC model
\citep[after generalizing the one dimensional model presented by][to the multivariate case]{Blomberg:2003eg,Harmon:2010gh}

\be
\begin{array}{rcl}
\vec{\omega}_{i} & = & \vec{0}, \\
\mathbf{\Phi}_{i} & = & \mathbf{I}, \\
\mathbf{V}_{i} & = & \int_{t_{i}^{s}}^{t_{i}^{e}}e^{s\mathbf{R}_{i}}\mathbf{\Sigma}_{i} \mathbf{\Sigma}_{i}^{T} e^{s\mathbf{R}_{i}^{T}}\ud s.
\end{array}
\ee
See Eq. \eqref{eqmvOUmoments} for how to calculate the integral, for $\mathbf{V}_{i}$, when the matrix $\mathbf{R}$ is eigendecomposable.
\item BM with drift

\be
\begin{array}{rcl}
\vec{\omega}_{i} & = & \vec{h}_{i}t_{i}, \\
\mathbf{\Phi}_{i} & = & \mathbf{I}, \\
\mathbf{V}_{i} & = & \mathbf{\Sigma}_{i} \mathbf{\Sigma}_{i}^{T}t_{i} .
\end{array}
\ee
\item BM with trend---in the most general setup of a linear form under the integral for $\mathbf{V}_{i}$
\citep[based on the Supporting Information of][in the one dimensional case]{Harmon:2010gh}

\be
\begin{array}{rcl}
\vec{\omega}_{i} & = & \vec{0}, \\
\mathbf{\Phi}_{i} & = & \mathbf{I}, \\
\mathbf{V}_{i} & = & \int^{t_{i}^{s}+t_{i}}_{t_{i}^{s}}\left(\mathbf{U}s+\mathbf{W}\right) \ud s = \mathbf{U}\frac{t_{i}^{2}}{2}+\mathbf{W}t_{i}.
\end{array}
\ee
\item The white noise process corresponds to a situation, where \VM{
the phylogeny does not contribute to the covariance structure between the species, so that 
the all species are regarded as independent identically distributed observations of the same 
multivariate Gaussian distribution with global mean $\vec{x}_{0}=\vec{\mu}$ and same 
variance--covariance matrix $\mathbf{\Sigma_e}$.}
\end{enumerate}
Naturally everything should be appropriately (as described in Section \ref{sbsecNAs}) adjusted if
missing values are present.

\section[Technical correctness]{Technical correctness}\label{secPostQuant}
Validating the technical correctness is an important but often neglected step in the development of 
likelihood calculation software. This step is particularly relevant for complex multivariate models, 
because logical errors can occur in many levels, such as the mathematical equations for the different 
terms involved in the likelihood, the programming code implementing these equations, the code responsible 
for the tree traversal, the parametrization of the model and the preprocessing of the input data. 
These logical errors add up to numerical errors caused by limited floating point precision, which 
can be extremely hard to identify. Ultimately, these errors lead to wrong likelihood values, 
false parameter inference and wrong analysis. All these concerns motivate for a systematic 
approach of testing the correctness of the software. 

We implemented a technical correctness test of the three models currently implemented in \pkg{PCMBase} 
using the method of posterior quantiles proposed by \citet{Cook:2006km}. The posterior quantiles method 
(Alg. \ref{algPostQuant}) 
is a simulation based approach. It employs the fact that, for a fixed prior distribution of the model parameters, 
the sample of posterior quantiles of any model parameter, $\theta$ 
is uniform \citep[see e.g.][for details]{Cook:2006km,Mitov:2017eg}.
Thus, any deviation from 
uniformity of the posterior quantile sample for any of the model parameters indicates the presence of an error, 
either in the simulation software, or in the likelihood calculator used to generate the posterior samples. 

\TScom{put heading to the procedure.}
\KBcom{Venelin is this OK?}
\begin{algorithm}[!htp]
\caption{Posterior quantiles method}\label{algPostQuant}
\begin{algorithmic}[1]
\State Sample ``true'' parameters $\Theta$ from the prior; 
\State Simulate random data, $\mathbf{X}_{\Theta}$, under the model specified by $\Theta$;
\State Generate a sample $S_{\theta}$ from the posterior distribution 
$P_{\theta}=P(\theta\vert \mathbf{X}_{\Theta})$;
\State Calculate the empirical quantile of the ``true'' $\theta$ in $S_{\theta}$;
\end{algorithmic}
\end{algorithm}

We used a fixed non--ultrametric tree of $N=515$ tips with two regimes ``a'' and ``b''. The tree was generated 
using the functions \code{pbtree()} and \code{sim.history()} from the package \pkg{phytools}  \citep{Revell:2011ku}. 
We implemented the posterior quantile test using the \pkg{BayesValidate} \proglang{R}--package \citep{Cook:2006km}. 
For each model we set a parametrization and a prior distribution as follows:
\begin{itemize}
\item BM \linebreak 
\noindent 3 parameters: $\Theta_{BM}=[\Sigma_{11},\Sigma_{12},\Sigma_{e,11}]$, such that
\[{\mathbf{\Sigma} _a} = \left( {\begin{array}{*{20}{c}}
{\Sigma_{11}}&{\Sigma_{12}}\\
{\Sigma_{12}}&{\Sigma_{11}}
\end{array}} \right),\,{\mathbf{\Sigma} _b} = \left( {\begin{array}{*{20}{c}}
{\Sigma_{11}}&{0}\\
{0}&{\Sigma_{11}}
\end{array}} \right),\,{\mathbf{\Sigma} _{e,a}} = {\mathbf{\Sigma} _{e,b}} = \left( {\begin{array}{*{20}{c}}
{\Sigma_{e,11}}&{0}\\
{0}&{\Sigma_{e,11}}
\end{array}} \right)\]
\noindent prior: $\Sigma_{11}\sim\text{Exp}(1),\,\Sigma_{12}\sim\mathcal{U}(-0.9\Sigma_{11},\,0.9\Sigma_{11}),\,\Sigma_{e,11}\sim\text{Exp}(10)$.
\item OU \linebreak 
\noindent 8 parameters: $\Theta_{OU}=[\Theta_{BM},\theta_{b,1},\theta_{b,2},H_{b,11},H_{b,12},H_{b,22}]$, 
such that $\mathbf{\Sigma}_{a}$, $\mathbf{\Sigma}_{b}$, $\mathbf{\Sigma}_{e,a}$ and $\mathbf{\Sigma}_{e,b}$ 
are defined as for BM and
\[{\vec{\theta} _a} = \left( {\begin{array}{*{20}{c}}
{0}\\
{0}\end{array}}\right ),\,{\vec{\theta} _b} = \left( {\begin{array}{*{20}{c}}
{\theta_{b,1}}\\
{\theta_{b,2}}\end{array}}\right ),\,{\mathbf{H} _a} = \left( {\begin{array}{*{20}{c}}
{0}&{0}\\
{0}&{0}\end{array}}\right ),\,{\mathbf{H} _b} = \left( {\begin{array}{*{20}{c}}
{H_{b,11}}&{H_{b,12}}\\
{H_{b,12}}&{H_{b,11}}
\end{array}} \right)\]
\noindent prior: for parameters in $\Theta_{BM}$ the same prior has been used as for the 
BM model; for the new parameters, the prior has been set as  
$\theta_{b,1}\sim\mathcal{N}(1,.25),\,\theta_{b,2}\sim\mathcal{N}(2,.5),\,H_{b,11}\sim\text{Exp}(1),\,H_{b,22}\sim\text{Exp}(1),\,H_{b,12}\sim\mathcal{U}(-0.9\sqrt{H_{b,11}H_{b,22}},\,0.9\sqrt{H_{b,11}H_{b,22}})$.
\item JOU \linebreak 
\noindent 9 parameters: $\Theta_{JOU}=[\Theta_{OU},\Sigma_{j,11}]$, such that 
$\mathbf{\Sigma}_{a}$, $\mathbf{\Sigma}_{b}$, $\mathbf{\Sigma}_{e,a}$, $\mathbf{\Sigma}_{e,b}$, $\vec{\theta}_{a}$, $\vec{\theta}_{b}$, $\mathbf{H}_{a}$, $\mathbf{H}_{b}$ 
are defined as for OU and
\[{\vec{\mu} _{j,a}} = \left( {\begin{array}{*{20}{c}}
{-\theta_{b,1}}\\
{-\theta_{b,2}}\end{array}}\right ),\,{\vec{\mu} _{j,b}} = \left( {\begin{array}{*{20}{c}}
{\theta_{b,1}}\\
{\theta_{b,2}}\end{array}}\right ),\,{\mathbf{\Sigma} _{j,a}} = \left( {\begin{array}{*{20}{c}}
{1}&{0}\\
{0}&{1}\end{array}}\right ),\,{\mathbf{\Sigma} _{j,b}} = \left( {\begin{array}{*{20}{c}}
{\Sigma_{j,11}}&{0}\\
{0}&{\Sigma_{j,11}}
\end{array}} \right)\]
\noindent prior: for parameters in $\Theta_{OU}$ the same prior has been used as for the OU model; 
for the new parameter, the prior has been set as  $\Sigma_{j,11}\sim\text{Exp}(10)$.
\end{itemize}

For each, model, we ran the function \code{validate()} from the \pkg{BayesValidate} package, setting the number of 
replications to $48$. The results are summarized in Fig. \ref{figPostQuantiles}. All Bonferroni adjusted 
p--values of the absolute $Z_{\theta}$ statistics were above $0.2$, showing that the posterior quantiles did not 
deviate from uniformity \citep[see][for details on $Z_{\theta}$ statistic]{Cook:2006km}. 

\begin{figure}
\begin{center}
\includegraphics[width=0.6\textwidth]{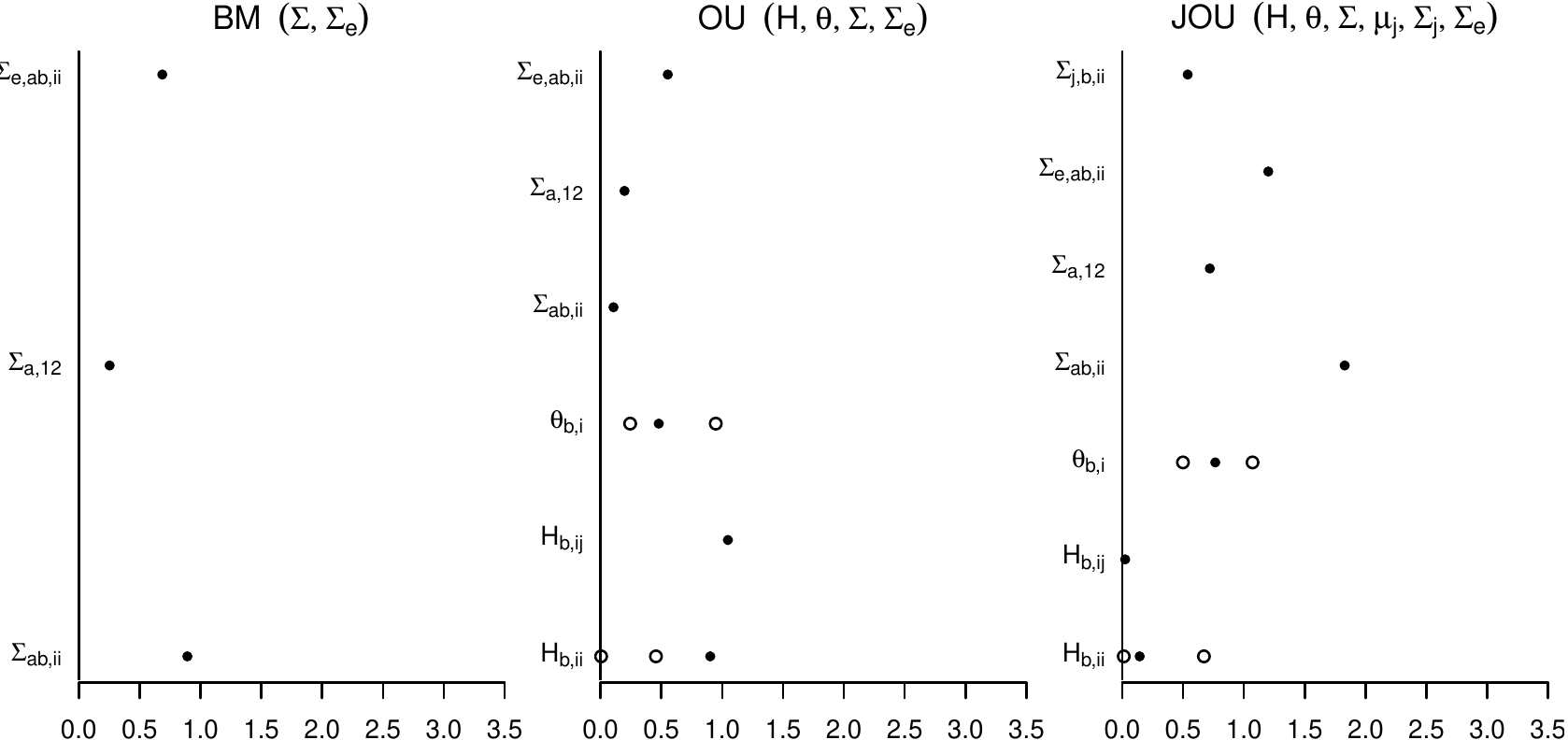}
\end{center}
\caption{Absolute 
$Z_{\theta}$--statistics for the four posterior quantile tests. 
The $Z_{\theta}$ statistic is described by \citet{Cook:2006km}. 
High values indicate deviation from uniformity of the posterior quantile distribution for an 
individual model parameter (circles) or a batch of several model parameters (bullets). The reported values, 
smaller than 3 for all parameters, had insignificant p-values as well as Bonferroni--adjusted p--values. 
The plots were generated using the package \pkg{BayesValidate} \citep{Cook:2006km}.} 
\label{figPostQuantiles}
\end{figure}

\section[Discussion]{Discussion}\label{secDiscussion}
Currently the mathematical frameworks proposed for PCMs are applied to situations
that are very different from the original motivation of a between species analyses \TS{within a small clade} 
of some quantitative trait. They \TS{are employed in many situations with a} tree structure behind the measurements.
For example, \del{the} traits being gene expression levels 
\citep{Bedford:2009pnas,Rohlfs:2013bl} or epidemiological \TS{measurements}
\citep[the tree connects the epidemic's outbursts][]{OPybus:2012pnas} are analysed.
\TScom{Since very recently the tree structure is being relaxed to a network
structure \citep{Jhwueng:2015eb} 
or even a structure with ``tunnels'' between branches where a continuous in time
exchange of signal takes place 
\citep{Bartoszek:2017ez,Drury:2016io, Manceau:2016kz, Nuismer:2015fd}.[this disrupts the flow here, 
move to further below?]}
\VMcom{I suggest to list these situations requiring fast likelihood calculations in one paragraph.} 
\KBcom{I agree but this is how it is done now, right?}
With large and diverse clades, 
there is a need to vary the parameters of the models across different clades or epochs in the tree.
Already e.g. \citet{Bartoszek:2012fw,Butler:2004ce,Hansen:1997ek}
showed the possibility of varying the deterministic optimum of OU processes.
\citet{Beaulieu:2012ex, Eastman:2011jf, Manceau:2016kz} went further to allow all parameters of 
the underlying SDE to vary over the tree. 
\TScom{do the studies up to here assume that we know the shift points? if so, say it.}
\VMcom{Eastman identifies shifts}
\del{A different problem is when the evolutionary shifts are unknown. } 
\TS{Estimating the time and branches for parameter changes has been proposed} \VM{ in }
\citet{Eastman:2011jf}, \citet{Ingram:2013cm}, \citet{Khabbazian:2016fr} and 
\citet{Bastide:2018eq} \TS{with implementations in} \pkg{AUTEUR}, \pkg{SURFACE}, 
\pkg{$l$1ou} and \pkg{PhylogeneticEM} \proglang{R} \del{packages for this} respectively. 

The situations mentioned above can easily 
require likelihood evaluations 
well beyond the amount of a ``standard'' optimization (e.g. 
an exponential number of regime patterns on the tree or reestimation to obtain the estimators' distribution).
Furthermore, the trees \TS{connected to such calculations} \TS{to be analysed} can be huge, going into thousands 
of tips \TS{such as HIV data analysed e.g. in \citep{Hodcroft:2014es,Mitov:2018co}}. \del{species.}
Hence, being able to quickly evaluate the likelihood is crucial. \VM{The \pkg{PCMBase} package offers this 
possibility for fast likelihood calculation for all models in the $\mathcal{G}_{LInv}$-family, 
including mixed--type models, where different types of models are realized on different parts of the phylogenetic tree.}
\TS{Further, it is extremely flexible 
allowing the user to  easily use it as a computational engine for their particular
modelling setup/parametrization. \pkg{PCMBase} is able to handle multiple standard extensions,
allowing the scientist to use all observed data. Finally, the package is written in such a way that it 
can be further developed to include more complex situations.}

\TS{Some ``standard extensions'' from Section \ref{secSpecialIssues} deserve special mention.
Firstly and briefly we remind the reader that \VM{\del{as}} \pkg{PCMBase} handles non--ultrametric
(Section \ref{sbsecNonUltra}) trees. Thus it  can directly use fossil data or pathogen data.}
\del{They can be inputted as a short extinct branches
from some point inside the tree} 
\KBtext{In the same Section \ref{sbsecNonUltra}, we notice that from the perspective of \pkg{PCMBase}
the out--degree of an internal node is irrelevant. This is as the likelihood is
calculated as the product over all daughter clades. Therefore, our computational
engine should be appreciated by users who have poorly resolved trees with polytomies.}
\VM{\del{\TS{Currently, it is assumed that all samples are tips in the tree, thus we do not support sampled ancestor trees 
\cite{Gavryushkina:2014fw}.
However, it will be straightforward to implement internal measurements on the tree if required by the user.}}}
\KBcom{How? I do not see anyway to do this by the user without modifying the package's code by them.
I have no problem with removing the below text, but I do not see why it is straightforward
for the end user.
}
\del{Even though from the mathematical perspective it is straightforward to have 
internal measurements on the tree we did not implement this possibility.
This is for two reasons.
Firstly, not to complicate the user interface and internal structure of the package. 
However, secondly from a biological point of view  there can never be any certainty whether the fossil 
is a direct ancestor of the given clade or a close relative (leading to an extinct lineage) and  exactly where 
to place it. 
Furthermore, including fossils as non--contemporary tips, offers the 
user the flexibility to try out different placements and check for robustness of results. }

\TS{ \pkg{PCMBase} handles incomplete observations of traits, meaning 
partially measured fossils do not pose any problem.}
\KBtext{As mentioned
in Section \ref{sbsecNAs} \pkg{PCMBase} distinguishes two types of missingness, unobserved trait (\code{NA})
and non--existing trait (\code{NaN}). From the perspective of the user this might seem like a mere formality.
However, from the perspective of the likelihood calculations it makes a profound difference.
Unobserved traits are integrated over, meaning that first $\vec{\omega}_{i}$, $\mathbf{V}_{i}$, $\mathbf{\Phi}_{i}$
are calculated as if all $k$ traits were present and only afterwards are appropriate entries/rows and columns removed.
The second case of non--existing traits is treated differently, 
$\vec{\omega}_{i}$, $\mathbf{V}_{i}$, $\mathbf{\Phi}_{i}$ are calculated taking into account that the 
trait vector at the given node is from a lower dimension (i.e. $\mathbf{A}_{i}$, $\vec{b}_{i}$,
$\mathbf{C}_{i}$, $\vec{d}_{i}$ and $\mathbf{E}_{i}$ are taken from lower dimensions by removing 
appropriate entries/rows and columns).} 
\KBcom{VENELIN, please check that I did not miss anything up
and this is how it is implemented, if wrong please correct.}\VMcom{Up to that point, the paragraph is correct, 
but the following sentence is neither tested nor implemented, although it could be of interest. 
I suggest removing it for now.}\KBcom{If you do not feel that comfortable with possible future
extensions, then no problem please remove.}\del{
\KBtext{
This possibility allows the user to study the appearance of traits, \TS{i.e. when a trait appeared or 
disappeared on the tree.}}}  \del{ Because one removes the entries of the entries, rows
and columns prior to doing any of the ``moment calculations'' (see Def. \ref{defGLinv}), this indicates that
the trait(s) disappeared/appeared just after the ancestral node.}

Despite the generality, speed and easiness of use of the package the user has to be aware of a potential pitfall. 
Theorem \ref{thmpcmNorm} and  the proof of Thm. \ref{thmMvFast} indicate a numerical weakness of our method. 
If a branch ending at node $i$ is extremely 
short, then the associated with it variance--covariance matrix, $\mathbf{V}_{i}$, can be computationally singular. 
Hence, calculating its inverse, a necessary step to obtain the likelihood, will not be possible. 
\pkg{PCMBase} catches such an error and returns it, pointing to the offending node. 
As an alternative, it is possible to tolerate such an error: if the branch is shorter than a user-specified threshold 
(runtime options \code{PCMBase.Skip.Singular} and \code{PCMBase.Threshold.Skip.Singular}), 
the whole branch can be treated as a $0$--length branch and skipped during the likelihood calculation. 

All models above assume that the trait evolves on the tree structure, i.e. does not influence the 
branching pattern. However, there are tools such as \pkg{diversitree} which assume that the trait determines 
branching rates. 
Those approaches assume a single model and parameters across the whole tree though and it will be a future challenge 
to generalize them towards multiple regimes on a single tree, as we did here for the $\mathcal{G}_{LInv}$-models.

\section*{Acknowledgments}
V.M. and T.S. were supported ETH Zurich. K.B. was supported by the Knut and Alice Wallenbergs Foundation,
the G S Magnuson Foundation of the Royal Swedish Academy of Sciences 
(grant no. MG$2016$--$0010$) and is supported by the Swedish Research Council 
(Vetenskapsr\aa det) grant no. $2017$--$04951$. G.A. and T.S. thank ETH Zurich for funding.

\bibliography{References}
\bibliographystyle{model2-names}

\end{document}